\documentclass[a4paper, 11pt, twoside, openright]{Thesis}  

\usepackage[square, numbers, comma, sort&compress]{natbib}  
\usepackage{verbatim}  
\usepackage{vector}  
\usepackage[noend]{algpseudocode}
\usepackage{algorithm}
\usepackage[T1]{fontenc}
\usepackage{venturis2}
\usepackage{lmodern}
\usepackage{textcomp}    
\usepackage{amsfonts}
\usepackage{amsmath}
\usepackage{amsthm}
\usepackage{subcaption}
\usepackage{multicol}
\usepackage{tikz}
\usepackage{epstopdf}
\usepackage{epsfig}
\usepackage{mathrsfs}
\usepackage{amssymb}
\usepackage{microtype}   
\usepackage{ps4pdf}
\PSforPDF{
  \usepackage{pstricks}
}
\usepackage[compact]{titlesec}
\usepackage{booktabs}
\usepackage{sectsty}     
\usepackage{listings}
\usepackage{caption}

\allsectionsfont{\sffamily}
\numberwithin{algorithm}{chapter}
\hypersetup{urlcolor=blue, colorlinks=true}  

\begin{document}
\frontmatter    

\title  {Analysis of String Sorting using Heapsort}
\authors  {Igor Stassiy}

\maketitle


\setstretch{1.3}  


\pagestyle{fancy}  
\fancyhead[RE,LO]{\sffamily\bfseries\nouppercase{\rightmark}}
\fancyhead[LE,RO]{\thepage}
\vspace*{1cm}
\textbf{\large Hilfsmittelerkl\"arung}\\[1em]
Hiermit versichere ich, die vorliegende Arbeit selbst\"andig verfasst und keine anderen als die angegebenen Quellen und Hilfsmittel benutzt zu haben.
\\[0.3cm]

\textbf{\large Non-plagiarism Statement}\\[1em]
Hereby I confirm that this thesis is my own work and that I have documented all sources used.

Saarbr\"ucken, den 03.\ M\"arz 2014,\\[1.5cm]
\hspace*{1cm}(Igor Stassiy)\\[2cm]

\textbf{\large Einverst\"andniserkl\"arung}\\[1em]
Ich bin damit einverstanden, dass meine (bestandene) Arbeit in beiden Versionen
in die Bibliothek der Informatik aufgenommen und damit ver\"offentlicht wird.
\\[0.3cm]

\textbf{\large Declaration of Consent}\\[1em]
Herewith I agree that my thesis will be made available through the library of the Computer Science Department at Saarland University.

Saarbr\"ucken, den  03.\ M\"arz 2014,\\[1.5cm]
\hspace*{1cm}(Igor Stassiy)
\clearpage
\pagestyle{empty}  

\null\vfill\vfill\vfill\vfill\vfill
\textit{To my grandparents}

\begin{flushright}
- Igor
\end{flushright}

\vfill\vfill\vfill\vfill\vfill\vfill\null
\clearpage  

\pagestyle{empty}

\mbox{}
\clearpage
\setstretch{1.3}  

\addtotoc{Abstract}  
\abstract{
\addtocontents{toc}{\vspace{1em}}  

In this master thesis we analyze the complexity of sorting a set of strings.
It was shown that the complexity of sorting strings can be naturally expressed
in terms of the prefix trie induced by the set of strings. The model of
computation takes into account symbol comparisons and not just
comparisons between the strings. The analysis of upper and lower bounds for some
classical algorithms such as Quicksort and Mergesort in terms of such a model was shown.
Here we extend the analysis to another classical algorithm - Heapsort. We also
give analysis for the version of the algorithm that uses Binomial heaps as a
heap implementation.

\clearpage  
}

\pagestyle{empty}
\mbox{}
\clearpage
\setstretch{1.3}  

\acknowledgements{
\addtocontents{toc}{\vspace{1em}}  
I would like to thank my supervisor Prof. Dr. Raimund Seidel for suggesting the
question researched in this thesis and for numerous meetings that helped me move 
forward with the topic. With his expert guidance, more often than not, our
discussions ended in a fruitful step forwards towards completing the results to
be presented in this thesis.

I would also like to thank Thatchaphol Saranurak for
useful discussions on Binomial heapsort and for his valuable suggestions. My
gratitude also goes to Dr. Victor Alvarez for agreeing to review this thesis. }
\clearpage


\pagestyle{fancy}  

\tableofcontents

\setstretch{1.5}  
\clearpage  



\setstretch{1.3}  

\addtocontents{toc}{\vspace{2em}}  

\mainmatter   
\pagestyle{fancy}  


\chapter{Introduction}

\newtheorem{myob}{Observation}
\newtheorem{myde}{Definition}
\newtheorem{mythe}{Theorem}
\newtheorem{myle}{Lemma}
\newtheorem{myco}{Corollary}
\newtheorem{mycl}{Claim}

Sorting is arguably one of the most common routines used both in daily life
and in computer science. Efficiency or computation complexity of this procedure
is usually measured in terms of the number of ``key'' comparisons made by the algorithm.

While mainly the problem of sorting is considered solved, up to date there
are publications improving the efficiency of sorting when used on real-world
data or when used in a different computational model than the standard ``key
comparison'' model ~\cite{DBLP:conf/focs/HanT02,DBLP:conf/stoc/Han02}.

It is well known that the best one can achieve in the ``key comparison'' model 
of computation is $O(n\log(n))$ comparisons when sorting $n$ keys (more
precisely the lower bound is $\log_2(n!)$).
However when the keys to be sorted are strings, this complexity measure doesn't adequately 
reflect the running time of the sorting algorithm. The reason for this is that
comparing two strings lexicographically takes time proportional to the length of
their longest common prefix plus one and the model of counting key comparisons simply doesn't
take this fact into account. 

A more appropriate model of counting comparisons is to count the total number of
symbol comparisons made by the algorithm. It is however not clear what fraction
of time will long keys be compared, hence taking a long time, and how many
comparisons will involve short keys, hence taking only short time. 

\section{Previous work}

The core of the problem boils down to choosing parameters of the input
strings that well describe the running time of the algorithm on this input
instance ~\cite{DBLP:conf/soda/Seidel10}. If we choose only $n$ - the number of
strings, then by adding a large common prefix $P$ to all of the strings, we can force a large number of
symbol comparisons, as comparing two strings would require at least $|P|+1$
operations. If chosen just $m$ - the total length of all the strings, adding
a common suffix $S$ to all the strings the number of symbol comparisons stays
the same, while $m$ increases by $n S$

The problem of parametrization of sorting was dealt with in two ways. The first
was developing specialized algorithms for sorting strings. Some of these
algorithms perform $O(m+n\log(n))$ comparisons and can be argued to be worst-case optimal
\cite{DBLP:conf/soda/BentleyS97}. The other approach was to assume that the
input strings are generated from a random distribution and analyze standard
key-comparison algorithms using this assumption
~\cite{DBLP:conf/icalp/ValleeCFF09}.

Both approaches are unacceptable in certain cases: the first one - since
often standard sorting algorithms are used for string sorting (such as the sort
command in UNIX or in programming languages like C++, Java). The second approach
is unsatisfying because often the input data does not come from a random distribution.

In ~\cite{DBLP:conf/soda/Seidel10} it is argued that when sorting a set of $n$
strings $S$ using the lexicographical comparison, the entire prefix structure of
the strings (or the prefix trie) should be used to parametrize the complexity of 
sorting. Hence we can apply the term ``data-sensitive'' sorting as the running
time depends on the instance $S$.

The following definitions will help to arrive to the precise problem statement:

\begin{myde}
For a set of strings $S$ the prefix trie $T(S)$ is a trie that has all the
prefixes of strings in $S$ as its nodes. A string $w$ is a child of another
string $v$ iff $v$ is one symbol shorter than $w$ and $v$ is a prefix of $w$.
\end{myde}

\begin{myde}
Denote $\vartheta(w)$, the thickness of the node $w$ - the number of leaves
under the node $w$.
\end{myde}

Intuitively $\vartheta(w)$ is the number of strings that
have $w$ as a prefix. Note that the prefix trie $T(S)$ is not used as a
data structure but merely as a parametrization for the running time of the
algorithm. 

\begin{myde} 
Define the \emph{reduced trie} $\hat{T}(S)$ as the trie induced by the node
set \\ $\{w \in P(S) \: | \: \vartheta(w) > 1 \}$.
\end{myde}

\begin{figure}[h]
    \centering    
    \includegraphics[scale=0.5]{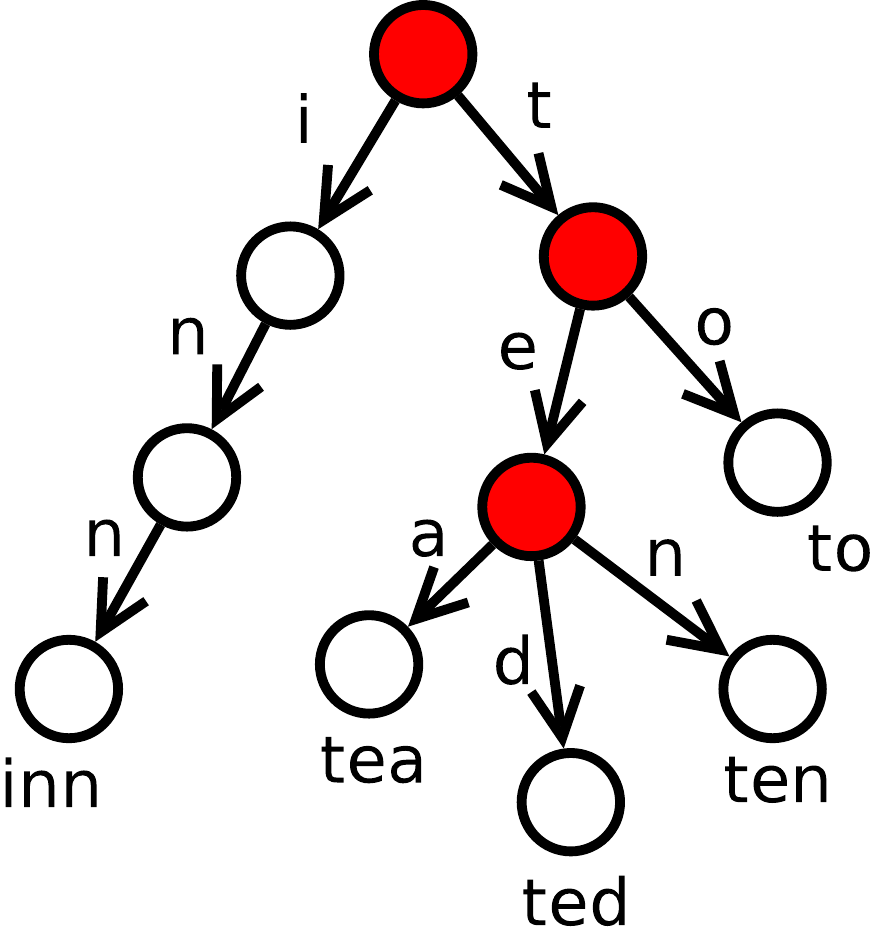}    
    \caption{A reduced trie example}
    \label{fig:reducedtrie}
\end{figure}

A \emph{reduced trie} simply discards strings which are suffixes of only
themselves. In ~\cite{DBLP:conf/soda/Seidel10} it is suggested that the vector
$(\vartheta(w))_{w \in P(S)}$ is a good parametrization choice and seems to
determine the complexity of sorting.

Here are the results presented in ~\cite{DBLP:conf/soda/Seidel10} that
are relevant to the results that are to be shown in this thesis: Denote $H_n =
\sum_{1 \leq i \leq n}1/i \approx 1+\log(n)$.
\begin{mythe}
\label{seideltheorem}
Let $Q(n) = 2(n+1)H_n-4n$ be the expected number of key comparisons performed by 
Quicksort on a set of $n$ keys.
Assume Quicksort is applied to string set $S$ employing the standard string
comparison procedure. Then the expected number of symbol comparisons performed
is exactly
\begin{equation}
\sum_{w \in P(S)} Q(\vartheta(w))
\end{equation}
\end{mythe}

The expression above is minimized when the trie $T(S)$ is highly balanced. Based
on this fact one can show that the expected number of comparisons is at least
$n\log^2(n)/\log(a)$. A similar result holds for Mergesort by simply
substituting $Q$ by $R$, where $R(n) = n\log_2(n)$. Note that this already
implies that it is \textit{suboptimal} to use a generic algorithm to sort
strings (as there are asymptotically better results
~\cite{DBLP:conf/soda/BentleyS97}), however this does not discourage the
analysis of such a parametrization, on the contrary, as most of the time exactly
the generic algorithms are used.
An advantage of such a parametrization is that it allows to compare sorting
algorithms relative to each other on a given input instance.

The essence of the proof of the theorem \ref{seideltheorem} (see proof of
theorem 1.1 in ~\cite{DBLP:conf/soda/Seidel10}) above is captured by the following lemma:
Call a comparison algorithm $f(t)$-\emph{faithful}, when sorting a set $X$ of
randomly permuted keys $x_1 < x_2 \ldots < x_N$, for any $0 < t \leq N$ and for
any ``range'' $X_{(i, i+t]} = \{x_{i+1}, \ldots, x_{i+t}\}$ of $t$
order-consecutive elements of $X$ the expected number of comparisons of the 
elements of $X_{(i, i+t]}$ is at most $f(t)$. Call it \emph{strongly} $f(t)$
faithful if this expected number of comparisons among elements $X_{(i, i+t]}$ is
\emph{exactly} $f(t)$.
\begin{myle}
If an $f(t)$-faithful (strongly $f(t)$ faithful) sorting algorithm is used to
sort a set S of randomly permuted strings using the standard string comparison
procedure, then the expected number of symbol comparisons is at most 
\begin{equation}
\sum_{w \in P(S)} f(\vartheta(w))
\end{equation}
\label{seidellemma}
\end{myle}
The lemma \ref{seidellemma} allows us to abstract from analysis of sorting
strings and instead focus on sorting general keys (in case of strings -
symbols).
The intuition behind the theorem, is that the algorithm does few comparisons overall if for a range of
the elements that appear consecutive in the order few comparisons is done.

\section{Problem statement}
The motivation for this thesis is to give a similar analysis for Heapsort.
Ideally we would like to get a result similar to the one in the theorem
\ref{seideltheorem}, namely, the expected number of symbol comparisons when
using Heapsort on a set of $n$ string $S$ to be of the form:

\begin{equation}
\sum_{w \in P(S)} B(\vartheta(w))
\end{equation}

Where $B(k)$ is the time required for Heapsort to sort $k$ keys. Note that the
result would only be interesting if the functions $Q(k), B(k)$ and $R(k)$ differ
by a constant factor. As only then we can compare the algorithms relative to
each other in an obvious way. Interest to Heapsort is well deserved since
it is an in-place asymptotically worst-case optimal algorithm for sorting in a
``key-comparison'' based model, unlike most versions of Quicksort and Mergesort.
Although the algorithm is asymptotically optimal, we will focus on getting the
right constants in the analysis.

In order to show this, we need to demonstrate that for any subrange of the
consecutive in-order elements few comparisons is done (ideally $n\log_2(n)$).

\section{Overview}

This thesis is divided into two main parts: analysis for Heapsort using the
standard binary heap data structure \ref{chapter:1} and the one with a binomial
heap~\cite{Vuillemin:1978:DSM:359460.359478} as a heap implementation
\ref{chapter:2}. In both section algorithms will be
broken into stages. Both parts will consider the stages of building the heap and
the actual sorting stage. In the beginning of each part we will show why
randomization is needed in order to achieve useful bounds on the number of
comparisons.

For \textit{Binary Heapsort} we will firstly consider a different model of
building the heap and prove bounds in this model, but later in the chapter the
standard way of building the heap will be analyzed, and tight coupling between
the two models will lead to the main result.

\chapter{Heapsort analysis}\label{chap2}
\label{chapter:1}
\nocite{*} 
\section{Overview}

Let us firstly recall the original Heapsort algorithm
~\cite{williams1964algorithm}.
The algorithm sorts a set, stored in an array, of $n$ keys using the binary heap
data structure.
Throughout this thesis, assume that all the heaps are \textit{max}-heaps and
that all logarithms are base $2$, unless otherwise mentioned.

We assume that the heap is implemented as an array with indexes starting from
$1$. The navigation on the heap can be accomplished with the following
procedures:

\begin{algorithmic}[1]
\Function{Left}{$a$}
    \State{\Return{2$a$}}
\EndFunction
\end{algorithmic}

\begin{algorithmic}[1]
\Function{Right}{$a$}
    \State{\Return{2$a$+1}}
\EndFunction
\end{algorithmic}

\begin{algorithmic}[1]
\Function{Parent}{$a$}
    \State{\Return{$a$/2}}
\EndFunction
\end{algorithmic}

\begin{myde}
The \textit{heap order} is maintained when the value at a node is larger or
equal to the values at the left and right children.
\end{myde}
Call a subheap the set of all the
elements under a particular node. As a subroutine to the algorithm, the
procedure \textit{MaxHeapify} is called on a node that has both its left and
right subheaps have heap order, but possibly not the itself node and its children. 
The procedure restores the heap order under a particular node:

\begin{algorithmic}[1]
\Function{MaxHeapify}{$A$,$i$}
    \label{alg:maxheapify}
    \State{$l \gets Left(i)$}
    \State{$r \gets Right(i)$}
    \If{$l \leq A.size \: {\bf and} \: A[l] > A[i]$}
      \State{$largest \gets l$}   
    \Else
      \State{$largest \gets r$}
    \EndIf
    
    \If{$r \leq A.size \: {\bf and} \: A[r] > A[largest]$}
      \State{$largest \gets r$}
    \EndIf
    \If{$largest \neq i$}
      \State{${\bf swap} \: A[i] \: with \: A[largest]$}
      \State{$MaxHeapify(A, largest)$}
    \EndIf    
\EndFunction
\end{algorithmic}

Having the \textit{MaxHeapify} procedure at hand, the \textit{BuildHeap} procedure
is straightforward: \textit{MaxHeapify} is called on all the nodes, starting from
the leaves of the heap in a bottom-up fashion. 

\begin{algorithmic}[1]
\Function{BuildHeap}{$A$}
    \State{$A.heap\-size = A.size$}
    \For{\textt{i = A.size \: \bf{downto} 1}}
      \State{$MaxHeapify(A, i)$}
    \EndFor
\EndFunction
\end{algorithmic}

Once the heap is built, we can run the actual Heapsort algorithm by
\textit{Popping} the root of the heap and recursively \textit{Sifting up} the
larger of roots of the left and right subheaps using the \textit{MaxHeapify}
procedure.

\begin{algorithmic}[1]
\Function{Heapsort}{$A$}
    \label{alg:heapsort}
    \State{$BuildHeap(A)$}
    \State{$A.heap\-size = A.size$}
    \For{\textt{i = A.size \: \bf{downto} 2}}
      \State{{\bf swap} A[1] \: with \: A[i]}
      \State{A.heapsize = A.heapsize-1}
      \State{$MaxHeapify(A, i)$}
    \EndFor
\EndFunction
\end{algorithmic}

In this thesis, we will give analysis for a version of the \textit{Heapsort}
algorithm with the Floyd's improvement ~\cite{DBLP:journals/cacm/Floyd64a}. In
~\cite{DBLP:journals/jal/SchafferS93} it is shown that this version is
average-case optimal even up to constants, unlike the original version of the 
algorithm. 

The modification of the algorithm lies in the \textit{MaxHeapify} procedure. Call
a node \textit{unstable} if it is currently being \textit{sifted down} (see
element $i$ in \ref{alg:maxheapify}). In the original version of the
algorithm, the children of the \textit{unstable} node are compared and the
larger one is compared to the \textit{unstable}, hence yielding $2$
comparisons each time the \textit{unstable} element is moved one level down in
the heap.
Since (see line $5$ in \ref{alg:heapsort}) the \textit{unstable} element
is taken from the very bottom of the heap, it is likely that it will have to be
\textit{sifted} all the way \textit{down} the heap.

In the Floyd's modification of the algorithm, the root node is not replaced with
the right-most element, but instead after the maximal element is popped a
\textit{hole} is formed. During the \textit{SiftDown} the larger of the two children of the
\textit{hole} is promoted up. Once the \textit{hole} is at the bottom-most level
of the heap, we put the \textit{right-most} heap element to the \textit{hole}
position, and possibly \textit{sift} it up to restore the \textit{heap order}.

\setlength\columnsep{20pt}
\begin{multicols}{2}
\begin{algorithmic}[1]
\Function{SiftDown}{$A$,$i$}
    \label{alg:siftdown}
    \State{$l \gets Left(i)$}
    \State{$r \gets Right(i)$}
    \If{$r \leq A.heapsize \: {\bf and} \: A[l] < A[r]$}
      \State{${\bf swap} \: A[r] \: {\bf and} \: A[i]$}
      \State{${\bf return} \: SiftDown(A, r)$}
    \Else
      \If{$l \leq A.heapsize$}
        \State{${\bf swap} \: A[l] \: {\bf and} \: A[i]$}
        \State{${\bf return} \: SiftDown(A, l)$}
      \EndIf
      \State{${\bf return} \:i$}
    \EndIf    
\EndFunction
\end{algorithmic}
\columnbreak
\begin{algorithmic}[1]
\Function{SiftUp}{$A$,$i$}
    \label{alg:siftup}
    \If {$i \neq 1$}    
    \State{$p \gets Parent(i)$}
    \If{$A[p] < A[i]$}
      \State{${\bf swap} \: A[p] \: {\bf and} \: A[i]$}
      \State{$SiftUp(A, p)$}   
    \EndIf   
    \EndIf 
\EndFunction
\end{algorithmic}
\end{multicols}

The reason for efficiency of the modified version of the algorithm is that when
the \textit{SiftUp} is called, the element being \textit{sifted up} was a leave
of the heap and hence it is intuitively not likely to go high up the heap. On
the other hand, in the original version of the algorithm we need to make $2$
comparisons per level of recursion and we always need to descend to the leaves
of the heap.

\begin{algorithmic}[1]
\Function{MaxHeapifyFloyd}{$A$,$i$}
    \label{alg:maxheapifyfloyd}
    \State{$j \gets SiftDown(A, i)$}
    \State{${\bf swap} \: A[A.heapsize] \: {\bf and} \: A[j]$}
    \State{$SiftUp(A, j)$}
\EndFunction
\end{algorithmic}

For a more detailed analysis of the algorithm see
~\cite{DBLP:journals/jal/SchafferS93}.

The following part of the chapter is logically divided into two parts.
In the first part, we consider the process of building a heap and count the
number of comparisons this stage of the algorithm makes. In the second
part the sorting part of the algorithm will be analyzed.

\section{Randomization}

In order to analyze performance of the \textit{Heapsort} algorithm,
let us define the setting more rigorously: let
$X= [x_1, x_2, \ldots x_n]$ such that $x_i \leq x_{i+1}$ for all $1 \leq i < n$,
$X_{i,i+r-1} = [x_{i}, \ldots, x_{i+r-1}] \subseteq X$ be an order consecutive
subrange of $X$.

In short, we would like to count the number of comparisons between the
elements of $X_{i, i+r-1}$ (possibly independent of $i$) made by the
\textit{Heapsort} algorithm when executed on $X$. We will frequently use
notation $X_r$ (omitting the index $i$) as most the results apply to any $i$.

When we talk about comparisons, only the comparisons between the elements of the
set $X_r$ are accounted for. For convenience, let us denote the $X_r$ elements
as \textcolor{red}{red} elements and the elements $X \setminus X_r$ as
\textcolor{blue}{blue} elements.

Before diving into analysis, let us consider the following example: let
the set $X = [1, 2, \ldots, n]$ and the set $X_{2\left \lfloor{ \log(n) }\right
\rfloor +1} = [n-2\left \lfloor{ \log(n) }\right \rfloor, \ldots, n]$ for $n =
2^k-1$ for some $k$. In the example below $k$ is $4$.

\begin{figure}[h]
    \centering    
    \includegraphics[scale=0.5]{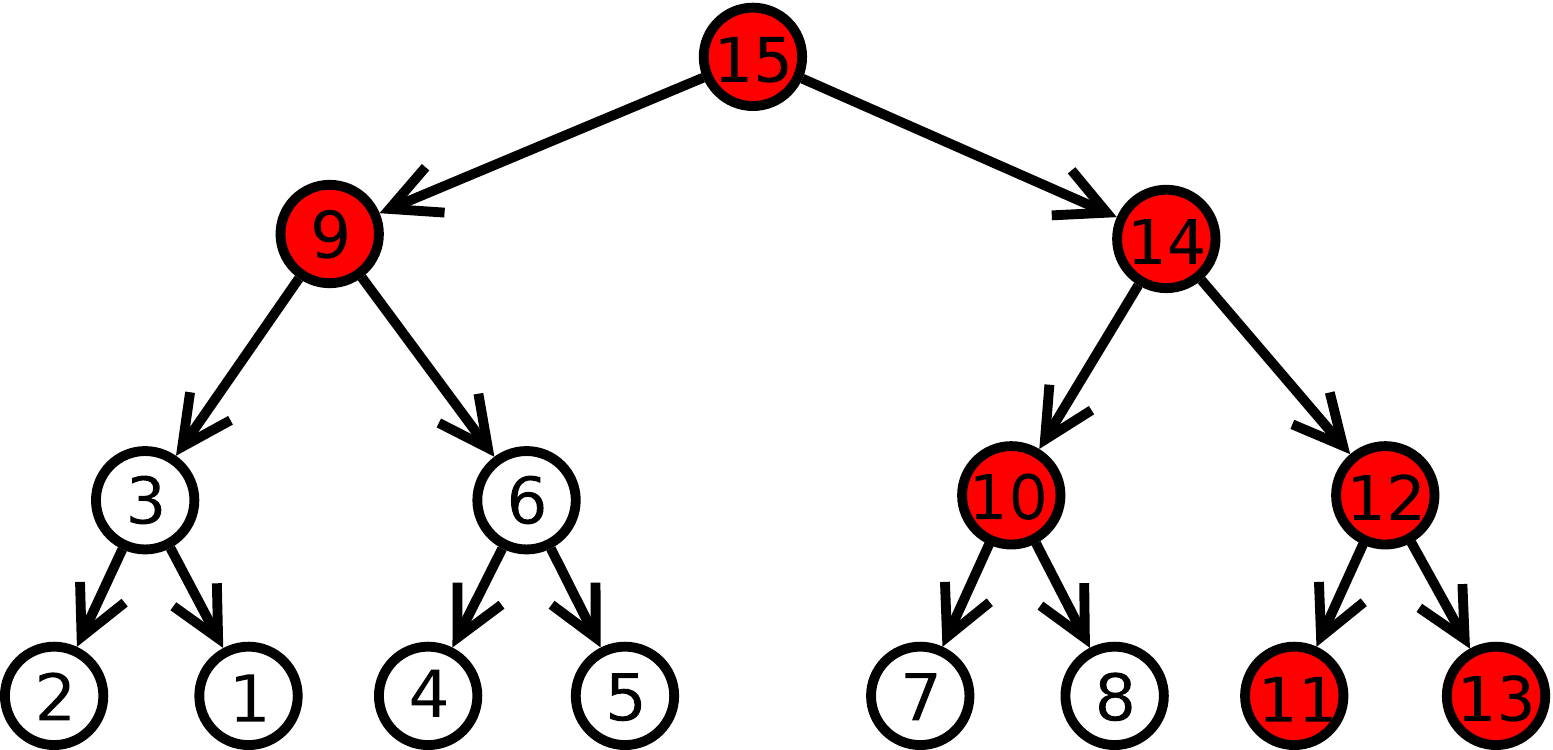}    
    \caption{A heap example. The set $X_r$ is in red.}
    \label{fig:exampleheap}
\end{figure}

It is not hard to see that we can place the elements $X_r$ as shown in the
example above, namely so that the elements of $X_r$ occupy the right spine of
the heap and each one (except for the root) have a left neighbor coming from
$X_r$. 

\begin{mycl}
The number of comparisons among elements of $X_r$ done by Heapsort
is $O(r^2)$.
\end{mycl}
\begin{proof}
The Floyd's modification algorithm would have to go to the end of the right
spine of the heap with every pop operation, hence causing $r/2$ operations on
the first \textit{sift down}, $r/2$ on the second and $\left \lfloor{r/2}\right
\rfloor-1$ on the third etc.
\end{proof}

The proposed solution is to randomly shuffle the initial array of the elements
$X$. It is well known that a uniform permutation of an array can be obtained in
$O(|X|)$ time (see \cite{DBLP:books/daglib/0023376} page $124$).
A ``uniform" heap is highly unlikely to have such a degenerate distribution of
the elements. A rigorous definition of a ``uniform'' heap is yet to come further
in the chapter.

\begin{myde}
A uniformly at random permuted array $A_{\sigma}$ is an array such that
every permutation of $A$ occurs equally likely in $A_{\sigma}$.
\end{myde}

Since the we are working with a uniformly random permutation of the initial
array, the bounds on the number of comparisons are \textit{expected}, although the
\textit{Heapsort} is a fully deterministic algorithm.

Another crucial aspect of the analysis is the so called ``randomness
preservation''. Intuitively, this property guarantees that given a
``random'' structure like heap, a removal of an element from the heap leaves the
heap ``random''. That is, after removal of an element, every heap configuration of a
heap is equally likely. We will give a more detailed insight into this property.
In the next chapter we demonstrate how the procedure \textit{BuildHeap} affects
the distribution of the elements in the heap and, in particular, show that the
procedure \textit{BuildHeap} is randomness preserving.

\section{Building the heap}
Let us show the effect of the \textit{BuildHeap} on the randomness of
distribution of elements $X$. Before we proceed we need to define what is a
random heap rigorously:
\begin{myle}
Let us call a heap uniformly
random if every valid heap configuration occurs with same probability.
\end{myle}
The claim is that if the original permutation is uniformly random, then the heap
after running \textit{BuildHeap} is also uniformly random, i.e. any heap
configuration on $X$ occurs with the same probability.

More specifically, we will now shot that the probability of a particular heap on
$m$ elements occurring is $1/H_m$, where $H_m$ is defined as:
\begin{eqnarray}
H_0 & = & H_1 = 1 \\
H_{m} & = & H_{k}H_{m-1-k}{m-1 \choose k}
\end{eqnarray} 
where $k, m-k$ is the number of elements in the left and right heaps
respectively. This would also imply that for a random heap, both its left and
right subheaps are uniformly random as well.

\begin{proof} Firstly, note that during heap construction, an element can only
move on its the path to the root of the heap or descend into the subheap it
was initially placed by the outcome of permutation of $X$.

Suppose an element $x$ is currently being \textit{sifted down}. Let $m$ be the
number of elements strictly below $x$ and $k/(m-k)$ be the number of elements
in the left/right subheaps respectively.
Lets show inductively on the size of the heap that after a \textit{SiftDown}
we obtain a uniformly random heap. The base of the induction is clear.

During a \textit{sift down} of $x$, $3$ cases can occur: element $x$ is swapped with one of its
children (right or left) stays on its position.

\begin{figure}[h]
    \centering    
    \includegraphics[scale=0.5]{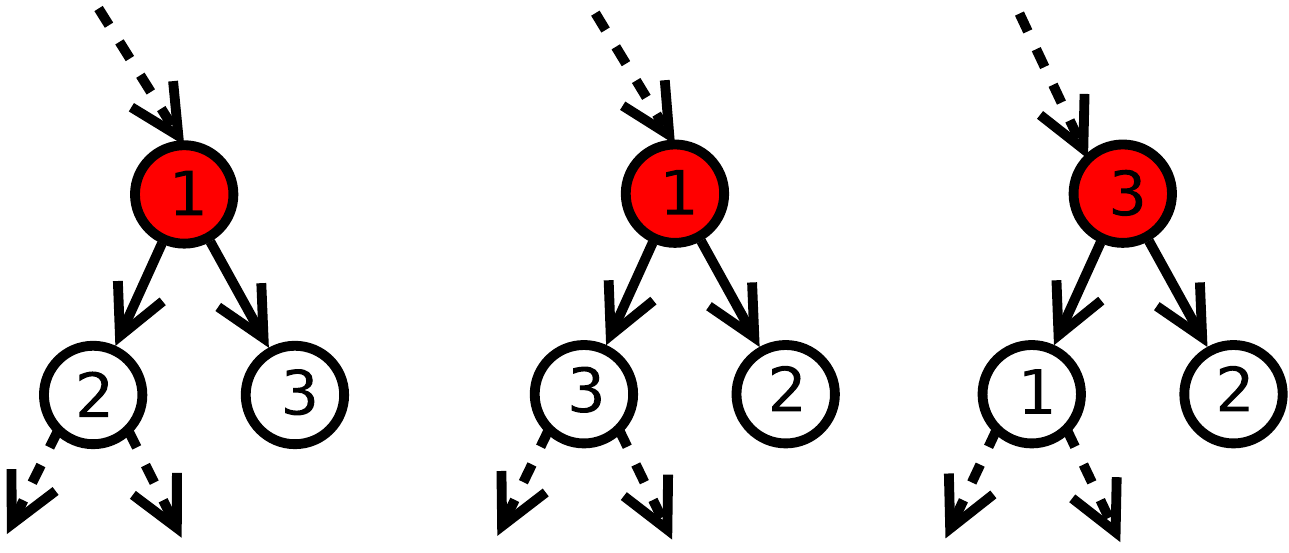}    
    \caption{3 cases leading to a different \textit{sift down} sequence: $x=1,
    \: x= 1, \: x = 3$}.
    \label{fig:siftdowncases}
\end{figure}

The probability that the maximal element in the heap of size $m+1$ is in the
left/right subheaps is $k/(m+1)$ and $(m-k)/(m+1)$ respectively. The probability
that the maximal element is $x$ itself is $1/(m+1)$.
Let $H_k$ be the number of possible heaps of size $k$ on a fixed set of
elements, then the probability of a particular heap occurring is $1/H_k$. Hence
the probability of a particular (left,right) pair of subheaps  is $1/(H_k
H_{m-k}{m \choose k})$, as choosing $k$ elements for the left subheap
completely determines the elements of the right subheap.

By induction hypothesis, it holds that the the probability of a particular
heap on $m+1$ elements (and by construction from procedure \textit{BuildHeap})
is:
\begin{equation*}
\frac{1}{m+1}\frac{1}{(H_{k}H_{m-k} {m \choose
k})}+ \frac{k}{m+1}\frac{1}{(H_{k}H_{m-k}{m \choose
k})}+\frac{m-k}{m+1}\frac{1}{(H_{k}H_{m-k}{m \choose k})}
\end{equation*}
Simplifying the expression gives 
\begin{equation}
H_{m+1} = H_{m-k}H_{k} {m \choose k}
\end{equation} 
which is exactly the induction hypothesis. Hence,
the inductive step holds and the proof is complete. A different proof can be
found in \cite[p.~153]{DBLP:books/aw/Knuth73}.
\end{proof}

Having the previous lemma in hand, we can proceed to analysis of the
asymptotic behavior of the expression describing the number of comparisons made
between elements $X_r$ during the \textit{BuildHeap} procedure.

One of the main results in this thesis is the following lemma:
\begin{myle} (Number of comparisons during construction of the heap) \\
The number of comparisons during construction of the heap between the elements
of subrange $X_r$ is $O(r)$
\end{myle}
\begin{proof}
For the sake of simplicity, let us assume that the number of elements in the
heap is of the form $m = 2^k-1$. The following argument extends to heaps of
different size, at the expense of making the proofs rather technical.

\begin{myde} (Expected number of comparisons caused by a single element)
Let $C(m, r)$ be the expected number of comparisons between the
\textcolor{red}{red} elements caused during a \textit{sift down} 
of the \textit{root} of a heap having $m$ elements and containing the
entire $X_r$.
\end{myde}

\begin{myde} (Expected number of comparisons in a heap)
$T(r, m)$ is the expected number of comparisons between the
\textcolor{red}{red} elements performed to build a subheap of size $m$,
given that elements of $X_r$ lie in the subheap
\end{myde}

Having the two definitions at hand, we can relate them as follows:

Let 
\begin{equation}
P_{2m+1, r, r'} = \frac{{2m+1-r \choose m-r'}{r
\choose r'}}{{2m+1 \choose m}}
\end{equation}

be the probability of a particular split of elements of $X_r$ to the left and
the right subheaps, such that there are $r'$ elements in the left heap and
$r-r'$ in the right one, then
\begin{equation}
T(2m+1, r) = \sum_{r \geq r' \geq 0}P_{2m+1, r, r'}(T(m,
r')+T(m, r))+C(2m+1, r)
\end{equation}

One thing to notice is that if the \textit{root} is in $X_r$, then the number of
elements to distribute to left and right subheap is $r-1$ and not $r$, as stated
in the recurrence relation. However, the following observation helps:

\begin{myob}(The number of comparisons is an increasing function in $r$) \\
\begin{equation}
T(2m+1, r) \geq T(2m+1, r') \: for \:  r \geq r'
\end{equation} 
\end{myob}
\begin{proof}
We can restrict counting the comparisons only for $r'$ largest elements among
the $r$ \textcolor{red}{red} elements.
\end{proof}

And hence we can think that there are in fact $r$ elements to distribute to the
subheaps.

To establish $C(m,r)$'s asymptotic behavior it is useful to observe that
\begin{myob}(Elements $X_r$ form connected subheaps)
On a path between two \textcolor{red}{red} elements, such that one is an
ancestor of the other, there are no
\textcolor{blue}{blue} elements. This implies that the \textcolor{red}{red}
elements form connected, possibly disjoint, subheaps after the
\textit{BuildHeap} procedure.
\end{myob}
\begin{proof}
Consider an element $x \in X_r$ in a heap. If there are is an element $x'\in
X_r$ ``under'' $x$, then there cannot be any elements $\bar{x} \in
X\backslash X_r$ on the internal path from $x$ to $x'$ as then this would
contradict the heap order.
\end{proof}

The previous observation suggests that an estimate for $C(m, r)$ is roughly
the ``average'' depth of a subheap formed by elements of $X_r$.  We can
formulate this relation as follows:

\begin{myle}(Expected number of comparisons during a \textit{sift down})
Let 
\begin{equation}
P_{2m+1, r, r'} = \frac{{2m+1-r \choose m-r'}{r
\choose r'}}{{2m+1 \choose m}}
\end{equation}

be the split probability in a heap of size $m$, containing entire $X_r$, such
that the left subheap contains $r'$ elements of $X_r$ and the right subheap
$r-r'$.

The expected number of comparisons created by an element \textit{ sifted down}
$C(m, r)$ can be upper bounded by
\begin{equation}
 C(2m+1, r) \leq \sum_{r > r' \geq 0} P_{2m+1, r-1, r'}(\frac{r'}{r-1}C(r',
 m)+\frac{r-r'-1}{r-1}C(r-r'-1, m))+2
\end{equation}
\hfil and
\begin{equation} 
C(2m+1, r) \leq 16\log(r)
\end{equation}
\end{myle}
\begin{proof}
Let an element $x \in X_r$ be \textit{sifted down}. By the observation above,
the elements of $X_r$ form connected subheaps. On its way down, $x$ will
enter one of those connected subheaps, and possibly cause one extra
comparison of two roots of the subheaps to decide which one to descend
to. The number of such comparisons is at most $O(r)$. When $x$ is
\textit{sifting down}, we need $2$ comparisons to decide if it should be placed lower
than the current position. Note that if $x \not \in X_r$, then at most one
comparison is needed.

We can now bound $C(m, r)$ assuming element $x$ descended to a subheap
containing $r$ elements, as any of the subheaps in which $x$ descends
has size at most $r$. Note that, when descending $x$ will swap with the smaller of its two
children. The probability that the smaller element is in the left or the right
subheap is $r'/(r-1)$ and $(r-r'-1)/(r-1)$ respectively, hence the expression
above.

Let us suppose that $C(r', m') \leq 16\log(r')$
for all $m' \leq m, r' \leq r$ and inductively show that $C(r, m) \leq
16\log(r)$. The base case is clear, as for $r \leq 2$ we can have only $1$
comparison of elements of $X_r$. 

It is not hard to show that $P_{m, r-1, r'}$ is increasing for $r' \leq (r-1)/2$
and decreasing for $r' \geq (r-1)/2$, the opposite holds for
$r'\log(r)+(r-r'-1)\log(r-r'-1)$. Applying the induction hypothesis to
the first equation in the statement of the lemma and the rearrangement
inequality:

\begin{eqnarray}
&&\sum_{r > r' \geq 0} P_{2m+1, r-1, r'}(\frac{r'}{r-1}C(r',
 m)+\frac{r-r'-1}{r-1}C(r-r'-1, m)) \leq \\
&&\sum_{r > r' \geq 0} P_{2m+1, r-1,
 r'}(\frac{16r'}{r-1}\log(r')+\frac{16(r-r'-1)}{r-1}\log(r-r'-1)) \leq \\
&&(\sum_{r > r' \geq 0} P_{2m+1, r-1, r'})(\sum_{r > r' \geq 0}
 \frac{16r'}{(r-1)^2}\log(r')+\frac{16(r-r'-1)}{(r-1)^2}\log(r-r'-1))
 \hspace{3em}
\end{eqnarray}

In the last inequality, there are $r$ terms, but we divide the
expression by $r-1$ instead of $r$, which only makes it larger. In the appendix
we show that, for $r \geq 16$:
\begin{equation}
\sum_{r > r' \geq 0 }r'\log(r') <
\frac{(r-1)^2}{2}\log(r)-\frac{(r-1)^2}{16}
\end{equation}
Notice that $\sum_{r > r' \geq 0} P_{2m+1, r-1, r'} = 1$. Combined with the
bound above, we get that for $r \geq 16$:
\begin{equation}
(12) \leq 16\log(r)-2
\end{equation} 

which is exactly what was needed to prove. Note that, for any $m$ and $r$, $C(m,
r) \leq r$ and so $C(m,r) \leq 16 \log(r)$ for $r < 16$ and the induction step
is complete. It is possible to reduce the constant $16$ further, as the expense
of increasing the complexity of analysis.
The point of the proof was to demonstrate that the number of comparisons can be
bounded by $C \log(r)$ for some constant $C$.
\end{proof}

Returning to the analysis of asymptotics of $T(m, r)$, what we wanted to
show is that $T(m, r) = O(r)$ for all $m$. We can show this inductively by
proving that
\begin{equation}
T(r, 2m+1) \leq c_1r-c_2\log(r)-c_3
\end{equation}

for some constants $c_1, c_2, c_3$, for all $m$ and $r$. Clearly, we can find 
such constants $c_i$ to make the base of the inductive claim hold. Suppose for
all $m' < m, r' < r$, the statement holds, then plugging the inductive
hypothesis to the statement of the lemma gives:

\begin{eqnarray}
T(2m+1, r) =  \sum_{r \geq r' \geq 0}P_{2m+1, r, r'}(T(m,
r')+T(m, r-r'))+C(2m+1, r) & \leq & \\ \sum_{r \geq r' \geq 0}P_{2m+1, r,
r'}(T(m, r')+T(m, r-r'))+16\log(r) & \leq & \\
\sum_{r \geq r' \geq 0}P_{2m+1, r,
r'}(c_1r-c_2(\log(r')+\log(r-r'))-2c_3)+16\log(r) & &
\end{eqnarray}

By definition, $\sum_{r \geq r' \geq 0} P_{2m+1, r, r'} = 1$. Then what we need
to show is that 

\begin{eqnarray}
\sum_{r \geq r' \geq 0}P_{2m+1, r,
r'}(c_1r-c_2(\log(r')+\log(r-r'))-2c_3)+16\log(r) & \leq & \\
c_1 r - c_2\log(r)-c3
\end{eqnarray}

Simplifying the expression gives

\begin{eqnarray}
-P_{2m+1, r, r'}(c_2(\log(r')+\log(r-r'))+2c_3)+16\log(r) & \leq &  \\ 
-c_2 \log(r)-c_3
\end{eqnarray}

Observe that both $P_{2m+1, r, r'}$ and $\log(r')+\log(r-r')$ are increasing for
$r' \in [1, r/2]$ and decreasing for $r' \in (r/2, r]$ hence we can apply the
rearrangement inequality
\begin{eqnarray*} 
\sum_{0 \leq r' \leq r} P_{2m+1, r, r'}(\log(r')+\log(r-r')) & \geq & \\
\frac{(\sum_{0 \leq r' \leq r} P_{2m+1, r, r'})}{r+1}\sum_{0 \leq r' \leq
r}(\log(r')+\log(r-r')) & = & \frac{2\log(r!)}{r+1}
\end{eqnarray*}

It is well known that 

\begin{equation}
e(\frac{n}{e})^n \leq n! \:\: \Rightarrow \:\: n\log(n)-n\log(e)+\log(e) \leq
\log(n!)
\end{equation}

Plugging the lower bound to the inequality $(22)$, leads to

\begin{eqnarray}
16\log(r) & \leq & \frac{2c_2}{r+1} (r\log(r)-r\log(e)+\log(e))-c_2\log(r)+c_3
\\
& = & (c_2 r\log(r)-2c_2 r\log(e)+2c_2\log(e))/(r+1)+c_3
\end{eqnarray}

Taking $c_2 = 16$ and $c_3 = 3c_2$ makes the statement hold and hence the
inductive step is complete. We have finally established that $T(m, r) = O(r)$ 
for any $m$.
\end{proof}

\section{Different model of building heap}

Coming to the second part of analysis, we need to count the number of
comparisons made by \textit{Heapsort} after the heap is built. Let us define the
concept of a ``premature'' heap.

\begin{myde}(Premature heap)
A heap is called \textit{premature} if none of the \textcolor{red}{red} elements
is at the top of the heap. We will call a heap \textit{mature} when it is not
\textit{premature} anymore.
\end{myde}

In case the largest \textcolor{red}{red} element is at the top of the heap, all
the \textcolor{red}{red} elements are still in the heap and form a connected
tree.

Let us count the expected number of comparisons once the heap is \textit{mature}
using a slightly different model of building heap. Namely, when a new element
$x$ is inserted into the current heap
\begin{itemize}
  \item if $x$ is larger than the current $root$, put $root=x$ and reinsert
  $root$
  \item in case $x$ is smaller or equal to $root$, insert $x$ into the left or
  right subheaps of the root with $p=1/2$.
\end{itemize}

We are about to show that the presented model of building the heap is randomness
preserving, unlike the model of building the heap using the
\textit{BuildHeap} procedure.

This model allows one to analyze the number of comparisons ``easily'' and, as
later will be shown, our previous model of heap construction via random
permutation of the initial sequence will have smaller expected number of comparisons of
elements in $X_r$.

Consider a heap on $N$ elements built in this way. Notice that a consequence
of this model is that a certain set of elements has probability $p_k^N = $$N-1
\choose k$ $/2^{N-1}$ of occurring in the left subheap of the root.

\begin{myle} (Preserving randomness) \\
After a \textit{PopMax} operation any heap on $N+1$ elements has the same
distribution as if it was built from scratch by taking all the elements out and
reinserting them. More precisely, the probability that a particular set lies in
the left subheap of the root is $p_k^N = $$N-1 \choose k$ $/2^{N-1}$.
\end{myle}
\begin{proof}
Let us call \textit{split} a certain set of the elements in the left subheap
(this set fully identifies the elements in the right subheap). Then $p^N_k$ is
the probability of a particular split on $N$ elements, such that there are $k$ elements in the left subheap. 
Such a split could have
resulted after a \textit{PopMax} operation on a heap with $N+1$ nodes in total
and

\begin{enumerate}
  \item $k$ nodes in the left subheap and the next largest element being in
the right subheap
\item $k+1$ nodes in the left subheap and the next largest element being in
the left subheap
\end{enumerate}

Then, in the first case the probability that the second largest element is in
the right subheap is $(N-k)/N$ similarly for the second case and the left
subheap the probability is $(k+1)/N$.

In
order to demonstrate randomness preservation, what we need to show is that

\begin{equation}
p^{N}_k =
\frac{k+1}{N}p^{N+1}_{k+1}+\frac{N-k}{N}p^{N+1}_k
\end{equation}

Rewriting the above equation, we have to show that $\displaystyle \frac{{N-1
\choose k}}{2^{N-1}} = \frac{k+1}{N}\frac{{N \choose
k+1}}{2^{N}}+\frac{(N-k)}{N}\frac{{N \choose k}}{2^{N}}$, or,
equivalently $\displaystyle \frac{{N-1 \choose k}}{2^{N-1}} = \frac{{N-1 \choose
k}}{2^{N}}+\frac{{N-1 \choose k}}{2^{N}}$, which clearly holds.
\end{proof}

\begin{myle}(Expected number of comparisons once the heap is \textit{mature})
The expected number of comparisons \textcolor{red}{red} elements once the heap
is \textit{mature} can be upper bounded by $r\log r$ under the model just described.
\end{myle}

\begin{proof}
Taking into account one of our previous observations, the elements of $X_r$
occupy a connected tree in the heap. As the root of the heap comes from $X_r$
and elements $X_r$ create a connected tree and hence the number of comparisons
between the \textcolor{red}{red} elements is not affected by the elements from
$X/X_r$.
Intuitively, we can remove all the elements in $X/X_r$ and both distribution of
the elements in $X_r$ and the number of comparisons between the
\textcolor{red}{red} elements will not be affected. Based on these facts we can devise a recursive
expression for the number of comparisons made between the elements of $X_r$:

Without loss of generality (based on the previous observation) we can take
$r=N$. Let $C_N$ be the number of comparisons for a subrange $X_N$ of size $N$,
created by an element \textit{sifted down}. Then by induction we can derive the
following:

\begin{enumerate}
  \item Clearly $C_0, \: C_1, \: C_2 = 0$.
  \item Once an element is $popped$ from the heap, we execute
  the recursive \textit{SiftDown} procedure. The two children of the root are
  compared only if both of the right or the left subheaps are non-empty. This
  happens with probability $1-p^N_{N-1}-p^N_0$.
  \item The probability that an element in the left subheap is larger than
  an element in the right subheap is $k/(N-1-k)$ (similarly we can derive the
  probability for the right subheap).
  \item Putting it all together we have: \\
  $\displaystyle C_N = 1-p^N_{N-1}-p^N_0 + \sum^{N-1}_{k=1} p^N_k (\frac{k}{N-1}
  C_k+\frac{N-1-k}{N-1}C_{N-1-k})$
\end{enumerate}
Because of the symmetry the expression simplifies to the following:
\begin{center}
$\displaystyle C_N = 1-p^N_{N-1}-p^N_0 + \sum^{N-1}_{k=1}$ $N-2 \choose k-1$
$C_k/2^{N-2}$
\end{center}
Instead of working with the recurrence directly, we will use an upper bound on
the $C_N$, namely,
\begin{center}
$\displaystyle C_N \leq 1+\sum^{N-2}_{k=0}$ $N-2 \choose k$
$C_{k+1}/2^{N-2}$
\end{center}
We will analyze the behaviour of this recurrence using the exponential
generating function \\ $\displaystyle g(z) = \sum_{k=0}^\infty C_k
\frac{z^k}{k!}$.
Multiplying the recurrence with $\displaystyle \frac{z^{N-2}}{(N-2)!}$ we have
\begin{eqnarray*}
\frac{C_N}{(N-2)!} z^{N-2}
& = &\frac{z^{N-2}}{(N-2)!}+\sum\limits^{N-2}_{k=0}{N-2 \choose k}
\frac{C_{k+1}}{(N-2)!}(\frac{z}{2})^{N-2} \\
& = & \frac{z^{N-2}}{(N-2)!}+\sum\limits^{N-2}_{k=0}
(\frac{z}{2})^{N-2}\frac{C_{k+1}}{k! (N-2-k)!} \\
& = & \frac{z^{N-2}}{(N-2)!}+\sum^{N-2}_{k=0}
\frac{C_{k+1}}{k!}(\frac{z}{2})^{k} \frac{1}{(N-2-k)!}(\frac{z}{2})^{N-2-k}
\end{eqnarray*}

Denoting $m = N-2$, we have for $m > 0$:

\begin{center}
$\displaystyle \frac{C_{m+2}}{m!}z^{m} = \frac{z^{m}}{m!}+\sum^{m}_{k=0}
\frac{C_{k+1}}{k!}(\frac{z}{2})^{k} \frac{1}{(m-k)!} (\frac{z}{2})^{m-k}$
\end{center}

Now, notice that $\sum^{m+1}_{k=1}
\frac{C_k}{(k-1)!}(\frac{z}{2})^{k-1} = g'(z/2)$ and $\sum^{m}_{k=0}
\frac{C_{k+1}}{k!}(\frac{z}{2})^{k}\frac{1}{(m-k)!}(\frac{z}{2})^{m-k}$ is the
convolution of $g'(\frac{z}{2})$ and $e^{\frac{z}{2}}$. Summing up all terms for
all $m$ on the left and right sides and taking into account that $C_2 = 0$, we
have $g''(z) = e^z-1+e^{z/2} g'(\frac{z}{2})$.

Lets now show that $g'(z) \leq g''(z)$ or equivalently $\sum_{k=0}^\infty
C_{k+1} \frac{z^k}{k!} \leq \sum_{k=0}^\infty C_{k+2} \frac{z^k}{k!}$. This
would clearly follow from the fact that $C_k$ is increasing. 

\begin{myle}
(The sequence $C_k$ is increasing)
\end{myle}

\begin{proof}
Suppose that $C_k$ is increasing for $k < N$. Let us show that $C_N > C_{N-1}$: 
\begin{eqnarray*}
C_N & = & 1+\sum_{k = 0}^{N-2}C_{k+1} \frac{{N-2 \choose k}}{2^{N-2}} \\ 
C_{N-1} & = & 1+\sum_{k = 0}^{N-3}C_{k+1} \frac{{N-3 \choose k}}{2^{N-3}}
\end{eqnarray*}

Subtracting the two expressions we get

\begin{eqnarray*}
C_N-C_{N-1} & = & C_{N-2}+\sum\limits_{k = 0}^{N-3}(\frac{{N-2 \choose
k}}{2^{N-2}}-2 \frac{{N-3 \choose k}}{2^{N-2}})C_{k+1} \\
& = & C_{N-2}+\sum_{k = 0}^{N-3}(\frac{{N-3 \choose k-1}}{2^{N-2}} -\frac{{N-3
\choose k}}{2^{N-2}})C_{k+1}
\end{eqnarray*}

The sum $\sum\limits_{k = 0}^{N-3}({N-3 \choose k-1}-{N-3
\choose k})C_{k+1}$ is non-negative, as ${N-3 \choose k-1}-{N-3
\choose k}$ is symmetric around $\frac{N-3}{2}$ and is positive for $k \geq \lfloor \frac{N-2}{3} \rfloor$ and $\displaystyle C_k$ is
increasing by induction assumption for $\displaystyle k \leq n-2$.
Hence $\displaystyle C_N-C_{N-1} > 0$ and induction is complete.
\end{proof}

Coming back to the generating function identity $g''(z) = e^z-1+e^{z/2}g'(z/2)$,
we showed that we can lower bound $g''(z)$ with $g'(z)$, from which it follows
that the following inequality holds:
\begin{center}
$\displaystyle g'(z) \leq g''(z) = e^z-1+e^{z/2}g'(z/2)$
\end{center}

Let $h(z) = g'(z)$, iterating the above inequality we have: 
\begin{center}
$\displaystyle h(z) \leq e^z-1+e^{z/2}(e^{z/2}-1)+e^{z/2+z/4}(e^{z/4}-1)\ldots =
\sum_{j = 0}^{\infty} e^z-e^{z(1-1/2^j)} $
\end{center}

Thus, the exponential generating function $h(z)$ has coefficients $
\displaystyle H_k = \sum_{j = 0}^{\infty} (1-(1-1/2^j)^k)$. Expanding the
expression using the fact that $\displaystyle (1-\frac{1}{2^j})^k =
\exp(-k/2^j)(1+O(1/k))$ gives:

\begin{eqnarray}
H_k & = & \sum_{0 \leq j < \lfloor \log k \rfloor}
1-e^{-k/2^j}+\sum_{j \geq \lfloor \log k \rfloor}
1-e^{-k/2^j}+o(1) \\
& = & \lfloor \log k \rfloor - \sum_{0 \leq j < \lfloor \log k \rfloor}
e^{-k/2^j}+\sum_{j \geq \lfloor \log k \rfloor}
1-e^{-k/2^j} \\
& = & \lfloor \log k \rfloor - \sum_{j < \lfloor \log k \rfloor}
e^{-k/2^j}+\sum_{j \geq \lfloor \log k \rfloor}
1-e^{-k/2^j}+O(e^{-k})
\end{eqnarray}

As $\sum_{j < 0} e^{-k/2^j} \leq \sum_{j > 0} e^{-kj} = e^{-k}/(1-e^{-k})-1 <
e^{-k}$. Shifting the summations by $\lfloor \log k \rfloor$,

\begin{center}
$\lfloor \log k \rfloor - \sum_{j < 0}
e^{-k/2^{j+\lfloor \log k \rfloor}}+\sum_{j \geq 0}
1-e^{-k/2^{j+\lfloor \log k \rfloor}}+O(e^{-k})$
\end{center}

It is easy to establish that $- \sum_{j < 0}
e^{-k/2^{j+\lfloor \log k \rfloor}}+\sum_{j \geq 0}
1-e^{-k/2^{j+\lfloor \log k \rfloor}} = O(1)$. A more precise analysis of
this function is available in \cite{DBLP:journals/jal/SchafferS93}. \\
\end{proof}

What we have shown is that a \textit{sift down} takes $\log(r)+O(1)$ time,
assuming the split distribution described above, and after each \textit{PopMax}
operation the heap randomness is preserved.
Hence the total number of comparisons is $r\log{r}+O(r)$ under the model
described.

It is seemingly not easy to derive bounds for the number of comparisons of
elements in $X_r$ after the heap is build but before the heap is
\textit{mature}. Perhaps one could count the number of comparisons
after the heap is built but before the heap is \textit{mature}, which would 
complete our analysis. The analysis is however not futile,as we will use details
of the section further in the thesis to show our main result.

\section{Original way of building heap}

Let us count the number of comparisons after the heap is built in a different
way. For now assume that $\forall x \in X_r \:$ is larger than the median of the
entire range $X$. This assumption is implementable in the following way: after
the heap is built, fill the lower most heap level (or possibly even the next
level) with an element $x'$ which is strictly smaller than all the elements in
$X$. What we need to achieve is that the number of additional elements on the
lowermost level of the heap is larger than the size of $X$.

The assumption guarantees that when we possibly place the right most
minimal element of the heap to the empty hole created by the \textit{SiftDown}
and execute \textit{SiftUp}, we do not have additional comparisons of elements
in $X_r$, as the element \textit{sifted up} is smaller than elements in $X_r$.

One might notice that if we introduce $n$ dummy elements to the heap, this would
incur additional comparisons and we would be cheating in a way, as we only
count the comparisons of the \textcolor{red}{red} elements and ignore
comparisons of the dummy elements. However we will show later in this section,
that dummy elements introduce only $O(n)$ more comparisons overall to the
running time of the algorithm, which is ``acceptable''.

\begin{myob}
Consider a node $x$ in the heap. Suppose during a \textit{SiftDown} a
comparison of elements  $left(x),\: right(x) \in X_r$ happens, in which case one
is promoted. Then the total number of the elements $\in X_r$ in the subheaps
rooted at $left(x), \: right(x)$ decreases by one.
\end{myob}

Also, note that such a comparison can happen iff initially there were some
elements $\in X_r$ in subheaps rooted at $left(x), \: right(x)$. Hence, an element
$x$ should be ``splitting'' a subset of $X_r$ into two subsets of size greater
than zero. 
\begin{myob}
There can be at most $r-1$ such positions $x$. Such
positions can be identified by taking the union of all the lowest common
ancestors of the pairs of elements of $X_r$ in the heap.
\end{myob}

Using the observation above we can express the maximum number of comparisons
in terms of ``splits'' of the set $X_r$ across the heap as follows:

\begin{myle}
Consider a node $x$ of the heap. Let it have $p$ and $q$ elements of $X_r$ in
the left and the right subheaps respectively. Then there can be at most $p+q-1$
comparisons of elements in $X_r$ at positions $left(x), \: right(x)$
\end{myle}

\begin{proof}
Using the observation above, after each comparison, the number of elements of
$X_r$ in the subheaps of $x$ decreases and we need at least one element of
$X_r$ in both right and left subheaps.
\end{proof}

\begin{myob}
The elements $x \in X_r$ stay on their original paths to the root of the heap 
during the \textit{SiftUp} procedure.
\end{myob}
\begin{proof}
Now the assumption that the elements $x \in X_r$ are larger than the median of
the array $X$ becomes handy. Note that during \textit{SiftUp} procedure, the elements 
only move up the path to the root of the heap and only the right-most element is
possibly moved from its path to the root. With the assumption, elements of $X_r$ do not appear in
the lower-most level of the heap. Hence, the elements $x \in X_r$ stay on their
original paths to the root of the heap. From now on we will assume that the elements $X_r$ do not lie in the lower most
layer of the heap during running of the \textit{Heapsort} algorithm, unless 
explicitly stated otherwise.
\end{proof}

\begin{myde} (Number of comparisons $C(m, r)$) Let $C(m, r)$ be the number of
comparisons between the elements $x \in X_r$ in a heap of size $m$ during the
sorting phase of the algorithm.
\end{myde}
Now, let us classify the comparisons of elements in $X_r$ to make the
computations easier.
\begin{myde} (Good comparisons) Call a comparison ``good'' if it compares the
roots of two subheaps of elements $\in X_r$ and call all the other comparisons
``bad''.
\end{myde}
\begin{myle} (Bounding ``good'' comparisons) \\
There are at most $r-1$ ``good comparisons'', where $r = |X_r|$
\end{myle}
\begin{proof}
Each time a ``good'' comparison happens, two subheap are merged. Once the
subheaps of elements in $X_r$ are merged, they stay connected. We
can merge at most $r$ separate subheaps and hence make at most $r-1$ ``good'' comparisons.
\end{proof}

From now on, we will only count ``bad'' comparisons. As ``good'' comparisons do
not contribute to the asymptotic behavior of the number of comparisons more
than a linear term.

\begin{myob}
Let there be $p$/$q$ elements from $X_r$ in the left/right subheaps of an
element $x$ respectively. Then there can be at most $p+q-1$ comparison at the
node $x$, that is comparing elements at positions $left(x), \: right(x)$.
\end{myob}

\begin{myle} (Upper bounding ``bad'' comparisons) The number of comparisons
$C(m, r)$ can be upper bounded by $r \log r+O(r)$
\end{myle}
\begin{proof}
Lets consider all the possible splits of the set $X_r$ by the current \textit{root} of
the heap. As was shown above, there can be at most $r-1$ ``bad" comparisons of
elements $X_r$ at a node, given that the node ``splits" the set $X_r$,
otherwise there will be no comparisons. Using the distribution of the elements
of $X_r$ in the heap:

\begin{center}
$\displaystyle C(2m+1, r) \leq r-1+\sum_{r' = 0}^{r}\frac{{2m+1-r \choose
m-r'}{r \choose r'}}{{2m+1 \choose m}}(C(m, r')+C(m, r-r'))$
\end{center}

The recursive relation reminds the formula for counting the expected number of 
comparisons under a different model of building the heap discussed above. Now, let us exploit the relation which we analyzed for a different model of
building the heap. We have shown that for 
\begin{center}
$\displaystyle C(0), C(1), C(2) = 0$ and $C(N) = 1+\sum {N-2 \choose
k}C(k+1)/2^{N-2}$
\end{center}
it holds that $C_{N} = \log(N)+O(1)$. Using a similar proof one can show that
for a relation
\begin{center}
$\displaystyle G(0), G(1), G(2) = 0$ and $G(N) = N+\sum {N \choose
k}(G(k)+G(N-k))/2^{N}$
\end{center}
it holds that $G(N) = N\log(N)-O(N)$.
In \cite{DBLP:journals/jal/SchafferS93} one can find a more precise analysis of the relation $G(N)$, in
particular it is shown that $G(N) < N\log(N)-\epsilon(N)N$, where $|\epsilon(N)|
< 1^{-3}$.

Coming back to the original expression, we wanted to show that
for
\begin{equation}
C(2m+1, r) = r-1+\sum_{r' = 1}^{r-1}\frac{{2m+1-r \choose m-r'}{r
\choose r'}}{{2m+1 \choose m}}(C(m, r')+C(m, r-r'))
\end{equation}

it holds that $C(m,r) \leq r\log(r)$. What we will try to show is that $C(m, r)
\leq G(r)$ which is an even a stronger statement.

Intuitively, the probability distribution $\displaystyle P_{m, r, r'} =
\frac{{2m+1-r \choose m-r'}{r \choose r'}}{{2m+1 \choose m}}$,
``encourages'' even splits of the elements in $X_r$ more than the split probability
$\displaystyle T_{m, r, r'} = {r \choose r'}/2^r$ (of $(G(r)$).
\begin{myle} (Split probability distributions)
There exists a range
$\displaystyle [r/2-\delta\ldots r/2+\delta]$ such that $P_{m, r, r'} \geq T(m,
r, r')$ for $r' \in \displaystyle [r/2-\delta\ldots r/2+\delta]$ and $P_{m, r,
r'} < T_{m, r, r'}$ otherwise.
\end{myle}
\begin{proof}
To show this more
rigorously, note that $\displaystyle \sum_{r'} P_{m, r, r'} = 1$ and
$\displaystyle \sum_{r'} T_{m, r, r'} = 1$, also, both $\displaystyle T_{m, r,
r'}$ and $\displaystyle P_{m, r, r'}$ are increasing on $\displaystyle r' \in
[1, \frac{r}{2}]$, however $\displaystyle R_{m, r, r'}$ is a pointwise product of
two binomial distributions. Note that if
$P_{m, r, r'} > T_{m, r, r'}$ then $\displaystyle \frac{{2m+1-r \choose
m-r'}}{{2m+1 \choose m}} > \frac{1}{2^r}$. And from the fact that $P_{m, r,
r'}$ is increasing on $[0\ldots r/2]$ and both $T_{m, r, r'}, P_{m, r, r'}$ sum
to $1$ it follows that there exists such a $\delta$.
\end{proof}
\begin{myco}
$P_{m, r, 'r} \geq T_{m, r, r'}$ for $r' = 0$
\end{myco}
The lemma above about the split probabilities combined with the next lemma
will give us a way to analyze the relation $C(m, r)$. Intuitively the next lemma
will show that the sequence $G(k)$ is ``concave'', i.e. the better the split of
the elements $X_r$, the less comparisons happen. As the split probability
$P_{m, r, r'}$ ``encourages'' even splits more than the probability $T_{m, r,
r'}$, we can bound $C(m, r)$ by $G(r)$.
\begin{myle} (``Concavity'' of the sequence $G(N)$) \\
For the sequence $G(N)$ defined above, it holds that $G(N-k)+G(k) \geq
G(N-k')+G(k')$, for $|N/2-k| > |N/2-k'|$
\end{myle}
\begin{proof}
Let us first show that $G(N)-G(N-1) \geq G(N-1)-G(N-2)$. Using the solution to
the recursion $G(N)$ from \cite{DBLP:journals/jal/SchafferS93}, namely
\begin{equation}
G(N) = N \sum_{j \geq 0} (1-\frac{1}{2^j})^{N-1}
\end{equation}
we can simplify $G(N+2)-2G(N+1)+G(N)$ to 
\begin{equation}
\ldots = \sum_{j \geq 0}
(-\frac{N+2}{2^{2j}}+\frac{2}{2^j})(1-\frac{1}{2^j})^{N-1}
\end{equation}
Now we need to prove that the sum non-negative for every $N$. To show the
non-negativeness, we will use the approximation of the sum with a definite
integral (see appendix).

Let $\displaystyle f(x) = \frac{1}{2^x}(1-\frac{1}{2^x})^{N-1}$, then
\begin{equation}
\int f(x) \mathrm{d}x = \frac{(1-\frac{1}{2^x})^{N}}{N \log_{e}(2)}
\end{equation}
Similarly, let $\displaystyle g(x) = \frac{1}{2^{2x}}(1-\frac{1}{2^x})^{N-1}$,
then using product rule, we have
\begin{equation}
\int g(x) \mathrm{d}x = 2^{-x}\frac{(1-\frac{1}{2^x})^{N}}{N
\log_{e}(2)}+\frac{(1-\frac{1}{2^x})^{N+1}}{N(N+1) \log_{e}(2)} =
\frac{(1-\frac{1}{2^x})^N}{N
\log_{e}(2)}-\frac{(1-\frac{1}{2^x})^{N+1}}{(N+1)\log_{e}(2)}
\end{equation}

Easy enough, $\sum_{i \geq 0}2f(i)-(N+2)g(i) $ is exactly the original sum.
Now, let us split the original sum into negative and positive terms. Looking at
multiple $\displaystyle -\frac{N+2}{2^{2j}}+\frac{2}{2^j}$ appearing in the
original summation, we conclude that terms for $i \geq \log(N+2)-1$ are non
negative. We have to be slightly careful here, as for $N+2 = 2^k$ for some
integer $k$ the term for $i = \log(N+2)-1$ is zero. Hence, to show the
non-negativity of the original sum, we can simply show that

\begin{equation}
\sum_{j \geq \log(N+2)-1}
(-\frac{N+2}{2^{2j}}+\frac{2}{2^j})(1-\frac{1}{2^j})^{N-1} \geq - \sum_{0 \leq j
< \log(N+2)-1} (-\frac{N+2}{2^{2j}}+\frac{2}{2^j})(1-\frac{1}{2^j})^{N-1}
\end{equation}

We are almost ready to apply the approximation of the sum by an integral.
It is easy to establish that $2f(x)-(N+2)g(x)$ is decreasing on $j \in
[0\ldots\log(N+2)-1)$ and similarly that $2f(x)-(N+2)g(x)$ is decreasing on $j \in
[\log(N+2)-1\ldots+\infty)$. Hence, we can lower bound r.h.s and upper bound
l.h.s as shown in the appendix by an integral and we would then need to show
that
\begin{equation*}
\sum_{j \geq \log(N+2)-1}
(-\frac{N+2}{2^{2j}}+\frac{2}{2^j})(1-\frac{1}{2^j})^{N-1}
\end{equation*}
\begin{eqnarray}
 & \geq &
\int_{\log(N+2)-1}^{\infty} 2f(x)-(N+2)g(x) \mathrm{d}x \\
\ldots & \geq & \int_{0}^{\log(N+2)-1} 2f(x)-(N+2)g(x) \mathrm{d}x \\
\ldots & \geq & -
\sum_{1 \leq j < \log(N+2)-1}
(-\frac{N+2}{2^{2j}}+\frac{2}{2^j})(1-\frac{1}{2^j})^{N-1} \\
\ldots & = & -\sum_{0 \leq j < \log(N+2)-1}
(-\frac{N+2}{2^{2j}}+\frac{2}{2^j})(1-\frac{1}{2^j})^{N-1}
\end{eqnarray}

We can rewrite the inequality $(17)$ as:

\begin{equation}
\int_{\log(N+2)-1}^{\infty} 2f(x)-(N+2)g(x) \mathrm{d}x \geq
\int_{0}^{\log(N+2)-1} (N+2)g(x)-2f(x) \mathrm{d}x
\end{equation}

or 

\begin{eqnarray}
2\frac{(1-\frac{1}{2^x})^{N}}{N\log_e2}-(N+2)(\frac{(1-\frac{1}{2^x})^N}{N
\log_{e}(2)}-\frac{(1-\frac{1}{2^x})^{N+1}}{(N+1)\log_{e}(2)})
\bigg{|}_{\log(N+2)-1}^{\infty} & \geq & \\
-2\frac{(1-\frac{1}{2^x})^{N}}{N\log_e2}+(N+2)(\frac{(1-\frac{1}{2^x})^N}{N
\log_{e}(2)}-\frac{(1-\frac{1}{2^x})^{N+1}}{(N+1)\log_{e}(2)})
\bigg{|}_{0}^{\log(N+2)-1}
\end{eqnarray}

Now, note that
$\displaystyle g(x)\bigg{|}_{x=0} = f(x) \bigg{|}_{x=0} = 0$ and
$\displaystyle f(x)\bigg{|}_{x\rightarrow \infty} = \frac{1}{N\log_e2}, g(x)
\bigg{|}_{x\rightarrow \infty} = \frac{1}{N\log_e2}-\frac{1}{(N+1)\log_e2}$.
We can get rid of the $\log_{e}2$ on both sides and rewrite the inequality as

\begin{eqnarray}
2\frac{1}{N}-(N+2)(\frac{1}{N}-\frac{1}{N+1}) = \frac{1}{N}-\frac{1}{N(N+1)}
 = \frac{1}{N+1}& \geq & \\
2(-2\frac{(1-\frac{1}{2^x})^{N}}{N}+(N+2)(\frac{(1-\frac{1}{2^x})^N}{N}-\frac{(1-\frac{1}{2^x})^{N+1}}{(N+1)})
\bigg{|}_{x = \log(N+2)-1}) & = & \\
2((1-\frac{1}{2^x})^N-(N+2)\frac{(1-\frac{1}{2^x})^{N+1}}{(N+1)})\big{|}_{x=\log(N+2)-1}& = & \\
2((1-\frac{1}{2^x})^N\frac{1}{2^x}-\frac{(1-\frac{1}{2^x})^{N+1}}{N+1}\big{|}_{x=\log(N+2)-1})
\end{eqnarray}

We can relax the inequality even further to 
\begin{equation}
\frac{1}{N+1} \geq
2((1-\frac{1}{2^x})^N\frac{1}{2^x}\big{|}_{x=\log(N+2)-1})
\end{equation}

To explore the behavior of a function $\displaystyle h(x) =
(1-\frac{1}{2^x})^N\frac{1}{2^x}$ notice that 
\begin{equation}
h'(x) =
(1-\frac{1}{2^x})^N(-\log_e2)\frac{1}{2^x}+N(1-\frac{1}{2^x})^{N-1}\frac{\log_e2}{2^{2x}}
\leq 0\end{equation}

is equivalent to 
\begin{equation}
N \leq (1-\frac{1}{2^x})2^x = 2^x-1 \leftrightarrow \log(N+1) \leq x
\end{equation}

And, hence the function is increasing for $\log(N+1) \leq x$ and attains its
maximum at $x = \log(N+1)$. Hence, 

\begin{equation}
2((1-\frac{1}{2^x})^N\frac{1}{2^x}\big{|}_{x=\log(N+2)-1}) <
2((1-\frac{1}{2^x})^N\frac{1}{2^x}\big{|}_{x=\log(N+1)})
\end{equation}

It is well known that $\displaystyle (1-\frac{1}{y})^y < \frac{1}{e}$ for $y
\geq 1$.
Hence, finally, we can relax the original inequality to

\begin{equation}
\frac{1}{N+1} > \frac{2}{e}\frac{1}{2^x}\big{|}_{x=\log(N+1)} =
\frac{2}{e}\frac{1}{N+1}
\end{equation}

which clearly holds. Now, in order to show that $G(N-k)+G(k) \geq
G(N-k')+G(k')$, for $|N/2-k| > |N/2-k'|$ let without loss of generality 
$N-k' > N-k > k > k'$, then what we need to show is
\begin{equation}
G(N-k')-G(N-k) \leq G(k)-G(k')
\end{equation}

but using telescoping sums we can rewrite it as 

\begin{eqnarray}
G(N-k')-G(N-k) & = & \\
\sum_{k' < i \leq k}(G(N-i)-G(N-i-1)) & \geq & \sum_{k' < i
\leq k}(G(i)-G(i-1)) \\ 
& = & G(k)-G(k')
\end{eqnarray}
which holds as the sequence $G(k)$ is convex. \\

As a sanity check, a simulating program for verifying ``concavity'' of $G(N)$
was written. The program verified the fact for $N \leq 10000$. The proof
actually shows us that $\displaystyle G(N+2)-2G(N+1)+G(N) >
\frac{1-\frac{2}{e}}{N+1}$, which is suggested by the results of the program as
well.
\end{proof}

Using the above, let us show that $C(m, r) \leq G(r)$
for all $m = 2^k-1$ inductively on $r$.
The base of induction clearly holds (for $r=0, 1, 2$ and for all $m$).
Clearly, $C(m, r) \leq G(r)$ for $m < r$. Using the fact that both $G(r)$ and
$C(m, r)$ are increasing in $r$ and the observation about the split probabilities we have

\begin{eqnarray}
& & C(2m+1, r) = \\
 & = & r-1 + \sum_{r' = 0}^{r}\frac{{2m+1-r \choose m-r'}{r
\choose r'}}{{2m+1 \choose m}}[C(m, r')+C(m, r-r')] \\
& \leq & r-1+\frac{2{2m+1-r \choose m}}{{2m+1 \choose m}}C(m,
r)+\sum_{r' = 1}^{r-1}\frac{{2m+1-r \choose m-r'}{r \choose
r'}}{{2m+1 \choose m}}[G(r')+G(r-r')] \\
& \leq & r-1+\frac{2{2m+1-r \choose m}}{{2m+1 \choose m}}C(m,
r)+\sum_{r' =
1}^{r-1}\frac{{r \choose r'}}{2^r}[G(r')+G(r-r')] \\
&\leq & r-1+\frac{2}{2^r}G(r)+\sum_{r' =
1}^{r-1}\frac{{r \choose r'}}{2^r}[G(r')+G(r-r')] \\
& = &r-1+\sum_{r' =
0}^{r}\frac{{r \choose r'}}{2^r}[G(r')+G(r-r')] = G(r)
\end{eqnarray}

the first inequality comes from the induction hypothesis $G(r') \geq C(m, r')$
for all $m$ and $r' < r$. The second follows from the fact that the sequence
$G(r')+G(r-r')$ is decreasing (from the proof above) for $r' \in [0, r/2]$ and
the observation about the split probabilities. The third inequality follows from
the observation about the split probabilities and induction hypothesis that
$C(m, r)\leq G(r)$.

Hence, $C(m, r)$ can be upper bounded by $r\log(r)$ as well for $m = 2^k-1$.
\end{proof}

Notice that the lemma above works only for the case where the initial number of
elements is of the form $2^k-1$. We can still make the \textit{Heapsort}
algorithm work as follows: 

\begin{myde}
Call an expansion of a number $n$ ``almost binary'' if 
\begin{equation}
n = \sum_k (2^{s_k}-1)
\end{equation}
and 
\begin{equation}
2^{s_k}-1 \geq \sum_{i < k} 2^{s_i}-1
\end{equation}
it follows that the sequence $s_1, s_2, \ldots, s_m$ is as small as possible and
$s_1 \geq s_2 \geq \ldots \geq s_m$
\end{myde}

It is easy to verify that such an expansion is unique for each $n$. But it could
be that $s_k = s_{k+1}$ for some $k$, for example when $n = 2(2^k-1)$. Notice
that, however, there is only one $k$ for which such a condition holds, and $s_k$
is the next to last in the sequence $\{s_k\}_{k}$. Because otherwise, let us
take the smallest such $i$, but then, if such an $k$ is not the last, $2^{s_k}-1
< 2^{s_{k+1}}-1+1$ and hence the second condition of the definition of the
``almost-binary'' expansion doesn't hold.

\begin{myle}
The number of terms in such an expansion is at most $\sum_{i = 1}\log^{(i)}(n)$
where $\log^{(i)}(n) = \log(n)$ for $i = 1$ and $\log^{(i)}(n) =
\log(\log^{(i-1)}(n))$ for $i > 1$.
\end{myle}
\begin{proof}
Denote the length of the almost-binary sequence of $n$ as $A(n)$. Consider
binary representation ${b_i}_i$ of $n$, which is know to have at most $\log(n)$ terms.
Then, $n-\sum_{i}2^{b_i} \leq \log(n)$ and hence, $A(n) \leq
A(\log(n))+\log(n)$. The statement of the lemma follows immediately.
\end{proof}

Having split the input elements to the subranges induced by the sizes of the
``almost binary'' expansion of $n$. As we have shown above, during \textit{
Heapsort} in each heap of size $2^{s_k}-1$ only $r\log(r)-O(r)$ comparisons are
created between the subrange elements. Suppose, that in a heap of size
$2^{s_k}-1$, there are $r_k$ \textcolor{red}{red} elements, then
the total number of comparisons in each heap is 
\begin{equation}
\sum_k r_k\log(r_k)-O(r_k) \leq r\log(r)-O(r)
\end{equation}

Once all of the heaps are sorted, we can merge then with a procedure analogous
to the one in the classic \textit{Mergesort}. We are only left to show that the
number of comparisons of the \textcolor{red}{red} elements during this procedure is $O(r)$.

The following observations will help to establish the result:

\begin{myob}
If the initial permutation is permuted uniformly at random, then each subrange
of size $2^{s_k}-1$ is also permuted uniformly, i.e. every subset of size
$2^{s_k}-1$ of input elements is equally likely to appear in the subrange.
\end{myob}

\begin{myob}
After a subrange of size $2^{s_k}-1$ is sorted, the \textcolor{red}{red}
elements lie contiguously. This observation is similar to the one above, which says that
elements of the range form a contiguous tree in a heap.
\end{myob}

Let us now count the number of comparisons during the merge:
\begin{myle}
To merge two arrays of size $r_k$ and $r_{k+1}$ we need at most $r_k+r_{k+1}-1$
comparisons of the \textcolor{red}{red} elements.
\end{myle}

\begin{myle}
During the merge procedure, there will be at most $\sum_{1 \leq k \leq m} kr_k$
comparisons of the \textcolor{red}{red} elements.
\end{myle}
\begin{proof}
$k$-th subrange's elements will be merged with the elements of subranges $1, 2,
\ldots k-1$ only.
\end{proof}

Denote $r_k$ the expected number of \textcolor{red}{red} elements in $i$-th heap.
In expectation, by the lemma above, we will have $E(\sum_{1 \leq k \leq m}
r_kk)$ comparisons. Exploiting linearity of expectation, we have 
\begin{equation}
E(\sum_{1 \leq k \leq m} kr_k) = \sum_{1 \leq k \leq m}k E(r_k) = \sum_{1 \leq k \leq m} k r
\frac{2^{s_k}-1}{n}
\end{equation}

We can rewrite the sum as a number of partial sums:

\begin{equation}
\sum_{1 \leq k \leq m} k r
\frac{2^{s_k}-1}{n} = \sum_{1 \leq k \leq m} \sum_{k \leq j \leq m} r
\frac{2^{s_k}-1}{n}
\end{equation}

And by the property of the sequence $s_k$ that 
\begin{equation}
2^{s_k}-1 \geq \sum_{i < k} 2^{s_i}-1
\end{equation}
we can conclude that
\begin{equation}
\sum_{1 \leq k \leq m} k r
\frac{2^{s_k}-1}{n} = \sum_{1 \leq k \leq m} \sum_{k \leq j \leq m} r
\frac{2^{s_k}-1}{n} \leq \sum_{1 \leq k \leq m} r (n/2^{k-1})/n \leq 2r
\end{equation}
And hence, during the merge procedure, the expected number of comparisons of the
\textcolor{red}{red} elements is $O(r)$ and the total number of comparisons
during \textit{Heapsort} between the \textcolor{red}{red} elements is
$r\log(r)+O(r)$.

As was promised, earlier, we give a proof that the dummy elements only
contribute $O(n)$ comparisons to the overall number of comparisons of the
algorithm:

There are $3$ possible comparison types:
\begin{enumerate}
  \item non-dummy and non-dummy
  \item dummy and non-dummy
  \item dummy and dummy 
\end{enumerate}

\begin{myob}
During the \textit{Delete Max} operation, a dummy element moves up on its path to
the root iff there was a comparison of two dummy elements.
\end{myob}

\begin{myob}If there was a comparison of a dummy and a non-dummy element, only
the non-dummy element moves up the path to the root of the heap, as dummies are strictly
smaller than non-dummy elements.
\end{myob}

Hence, we can bound the total number of comparisons dummy elements cause with
\begin{enumerate}
  \item cumulative path length of non-dummy elements to the root of the heap
  (comparisons of types $1, 2$)
  \item cumulative path length of dummy elements to their positions once
  \textit{Heapsort} is done (comparisons of type $3$)
\end{enumerate}

We have already shown that for a range of size $r$, there are $r\log(r)+O(r)$
comparisons between the \textcolor{red}{red} elements during the sorting procedure. Taking
$r=n$ we immediately get that comparisons of type $1$ account for
no more than $n\log(n)+O(n)$ of all the comparisons. But information theoretic
lower bound tells us that there are at least $n\log(n)+O(n)$ comparisons in any
algorithm in comparison-based model. Hence, there are only $O(n)$ comparisons of
type $2$. 

The following lemma will show that there are only $O(n)$ comparisons of type
$3$:

\begin{myle}
The cumulative path length of all the dummy elements to their positions in the
heap is $O(n)$
\end{myle}
\begin{proof}
Once \textit{Heapsort} is done, the number of dummy elements at depth
$\log(n)+1$ is $1$, at depth $\log(n)$ - $2$ and so on:
\begin{equation}
\sum_{0 \leq k \leq \lceil \log(n) \rceil} (\lceil \log(n) \rceil -k+1)2^k = O(n) 
\end{equation}
The detailed proof is omitted due to space constraints. Hence the cumulative
path length of the dummy elements to their final positions in the array is
bounded by $O(n)$.
\end{proof}

We can finally conclude that there are $O(n)$ comparisons of types $2, 3$ and
hence the dummy elements contribute only $O(n)$ comparisons overall.

\chapter{Binomial Heapsort analysis}\label{chap3}
\label{chapter:2}

Although often not a method of choice for sorting, it will be shown
that \textit{Binomial Heapsort}
~\cite{Vuillemin:1978:DSM:359460.359478} algorithm preserves randomness during
its runtime which makes it comparatively easy to analyze. The analysis will be carried out up to constants, as with the \textit{Binary HeapSort}.

To recap, a binomial heap is a forest of rooted heaps (trees), ordered by their
sizes and their roots stored in the \textit{Root list}. There is only one tree
for each size in the \textit{Root list}. For every tree $T$, $|T| = 2^k$ for
some $k$ and a tree of size $2^k,\: k > 0$ is a join of two trees $T_1, T_2$ of
size $2^{k-1}$, with the root of $T$ being the larger of the roots of $T_1,
T_2$. We can also view a \textit{Binomial heap} of size $2^k$ as a rooted tree,
with the root having $k$ children, each being a root of a heap of size $2^i, \:
0 \leq i \leq k-1$.

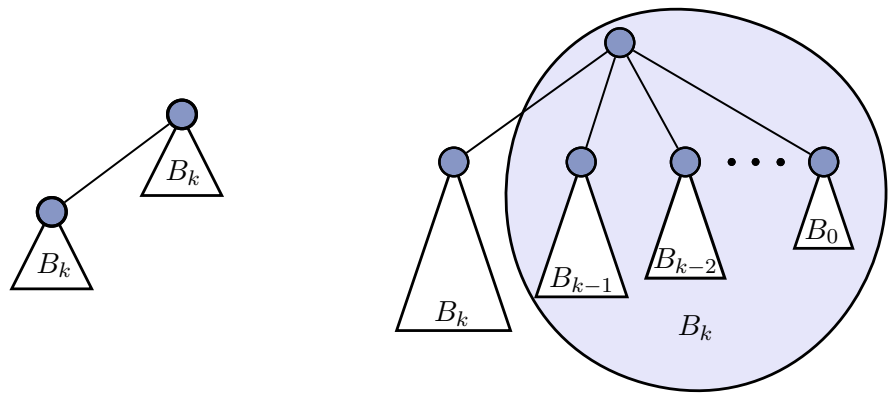
\begin{figure}
\begin{minipage}{0.5\textwidth}
\centering
\ifx\du\undefined
  \newlength{\du}
\fi
\setlength{\du}{15\unitlength}
\begin{tikzpicture}
\pgftransformxscale{1.000000}
\pgftransformyscale{-1.000000}
\definecolor{dialinecolor}{rgb}{0.000000, 0.000000, 0.000000}
\pgfsetstrokecolor{dialinecolor}
\definecolor{dialinecolor}{rgb}{1.000000, 1.000000, 1.000000}
\pgfsetfillcolor{dialinecolor}
\pgfsetlinewidth{0.070000\du}
\pgfsetdash{}{0pt}
\pgfsetdash{}{0pt}
\pgfsetbuttcap
\pgfsetmiterjoin
\pgfsetlinewidth{0.070000\du}
\pgfsetbuttcap
\pgfsetmiterjoin
\pgfsetdash{}{0pt}
\definecolor{dialinecolor}{rgb}{1.000000, 1.000000, 1.000000}
\pgfsetfillcolor{dialinecolor}
\fill (14.050000\du,11.150000\du)--(15.050000\du,13.150000\du)--(13.050000\du,13.150000\du)--cycle;
\definecolor{dialinecolor}{rgb}{0.000000, 0.000000, 0.000000}
\pgfsetstrokecolor{dialinecolor}
\draw (14.050000\du,11.150000\du)--(15.050000\du,13.150000\du)--(13.050000\du,13.150000\du)--cycle;
\pgfsetlinewidth{0.070000\du}
\pgfsetdash{}{0pt}
\pgfsetdash{}{0pt}
\pgfsetbuttcap
\pgfsetmiterjoin
\pgfsetlinewidth{0.070000\du}
\pgfsetbuttcap
\pgfsetmiterjoin
\pgfsetdash{}{0pt}
\definecolor{dialinecolor}{rgb}{1.000000, 1.000000, 1.000000}
\pgfsetfillcolor{dialinecolor}
\fill (17.300000\du,8.800000\du)--(18.300000\du,10.800000\du)--(16.300000\du,10.800000\du)--cycle;
\definecolor{dialinecolor}{rgb}{0.000000, 0.000000, 0.000000}
\pgfsetstrokecolor{dialinecolor}
\draw (17.300000\du,8.800000\du)--(18.300000\du,10.800000\du)--(16.300000\du,10.800000\du)--cycle;
\pgfsetlinewidth{0.080000\du}
\pgfsetdash{}{0pt}
\pgfsetdash{}{0pt}
\pgfsetbuttcap
\pgfsetmiterjoin
\pgfsetlinewidth{0.080000\du}
\pgfsetbuttcap
\pgfsetmiterjoin
\pgfsetdash{}{0pt}
\definecolor{dialinecolor}{rgb}{0.529412, 0.588235, 0.772549}
\pgfsetfillcolor{dialinecolor}
\pgfpathellipse{\pgfpoint{17.305446\du}{8.757772\du}}{\pgfpoint{0.360000\du}{0\du}}{\pgfpoint{0\du}{0.360000\du}}
\pgfusepath{fill}
\definecolor{dialinecolor}{rgb}{0.000000, 0.000000, 0.000000}
\pgfsetstrokecolor{dialinecolor}
\pgfpathellipse{\pgfpoint{17.305446\du}{8.757772\du}}{\pgfpoint{0.360000\du}{0\du}}{\pgfpoint{0\du}{0.360000\du}}
\pgfusepath{stroke}
\pgfsetbuttcap
\pgfsetmiterjoin
\pgfsetdash{}{0pt}
\definecolor{dialinecolor}{rgb}{0.000000, 0.000000, 0.000000}
\pgfsetstrokecolor{dialinecolor}
\pgfpathellipse{\pgfpoint{17.305446\du}{8.757772\du}}{\pgfpoint{0.360000\du}{0\du}}{\pgfpoint{0\du}{0.360000\du}}
\pgfusepath{stroke}
\pgfsetlinewidth{0.080000\du}
\pgfsetdash{}{0pt}
\pgfsetdash{}{0pt}
\pgfsetbuttcap
\pgfsetmiterjoin
\pgfsetlinewidth{0.080000\du}
\pgfsetbuttcap
\pgfsetmiterjoin
\pgfsetdash{}{0pt}
\definecolor{dialinecolor}{rgb}{0.529412, 0.588235, 0.772549}
\pgfsetfillcolor{dialinecolor}
\pgfpathellipse{\pgfpoint{14.052864\du}{11.211843\du}}{\pgfpoint{0.360000\du}{0\du}}{\pgfpoint{0\du}{0.360000\du}}
\pgfusepath{fill}
\definecolor{dialinecolor}{rgb}{0.000000, 0.000000, 0.000000}
\pgfsetstrokecolor{dialinecolor}
\pgfpathellipse{\pgfpoint{14.052864\du}{11.211843\du}}{\pgfpoint{0.360000\du}{0\du}}{\pgfpoint{0\du}{0.360000\du}}
\pgfusepath{stroke}
\pgfsetbuttcap
\pgfsetmiterjoin
\pgfsetdash{}{0pt}
\definecolor{dialinecolor}{rgb}{0.000000, 0.000000, 0.000000}
\pgfsetstrokecolor{dialinecolor}
\pgfpathellipse{\pgfpoint{14.052864\du}{11.211843\du}}{\pgfpoint{0.360000\du}{0\du}}{\pgfpoint{0\du}{0.360000\du}}
\pgfusepath{stroke}
\pgfsetlinewidth{0.050000\du}
\pgfsetdash{}{0pt}
\pgfsetdash{}{0pt}
\pgfsetbuttcap
{
\definecolor{dialinecolor}{rgb}{0.000000, 0.000000, 0.000000}
\pgfsetfillcolor{dialinecolor}
\definecolor{dialinecolor}{rgb}{0.000000, 0.000000, 0.000000}
\pgfsetstrokecolor{dialinecolor}
\draw (14.372485\du,10.970690\du)--(16.985826\du,8.998926\du);
}
\definecolor{dialinecolor}{rgb}{0.000000, 0.000000, 0.000000}
\pgfsetstrokecolor{dialinecolor}
\node[anchor=west] at (13.400000\du,12.550000\du){$B_k$};
\definecolor{dialinecolor}{rgb}{0.000000, 0.000000, 0.000000}
\pgfsetstrokecolor{dialinecolor}
\node[anchor=west] at (16.625000\du,10.225000\du){$B_k$};
\end{tikzpicture}
\end{minipage}%
\begin{minipage}{0.5\textwidth}
\centering
\ifx\du\undefined
  \newlength{\du}
\fi
\setlength{\du}{15\unitlength}
\begin{tikzpicture}
\pgftransformxscale{1.000000}
\pgftransformyscale{-1.000000}
\definecolor{dialinecolor}{rgb}{0.000000, 0.000000, 0.000000}
\pgfsetstrokecolor{dialinecolor}
\definecolor{dialinecolor}{rgb}{1.000000, 1.000000, 1.000000}
\pgfsetfillcolor{dialinecolor}
\pgfsetlinewidth{0.070000\du}
\pgfsetdash{}{0pt}
\pgfsetdash{}{0pt}
\pgfsetmiterjoin
\pgfsetbuttcap
\definecolor{dialinecolor}{rgb}{0.901961, 0.901961, 0.980392}
\pgfsetfillcolor{dialinecolor}
\pgfpathmoveto{\pgfpoint{28.074965\du}{7.025035\du}}
\pgfpathcurveto{\pgfpoint{34.574965\du}{8.975035\du}}{\pgfpoint{31.224965\du}{17.775035\du}}{\pgfpoint{25.574965\du}{16.175035\du}}
\pgfpathcurveto{\pgfpoint{19.924965\du}{14.575035\du}}{\pgfpoint{21.574965\du}{5.075035\du}}{\pgfpoint{28.074965\du}{7.025035\du}}
\pgfusepath{fill}
\definecolor{dialinecolor}{rgb}{0.000000, 0.000000, 0.000000}
\pgfsetstrokecolor{dialinecolor}
\pgfpathmoveto{\pgfpoint{28.074965\du}{7.025035\du}}
\pgfpathcurveto{\pgfpoint{34.574965\du}{8.975035\du}}{\pgfpoint{31.224965\du}{17.775035\du}}{\pgfpoint{25.574965\du}{16.175035\du}}
\pgfpathcurveto{\pgfpoint{19.924965\du}{14.575035\du}}{\pgfpoint{21.574965\du}{5.075035\du}}{\pgfpoint{28.074965\du}{7.025035\du}}
\pgfusepath{stroke}
\pgfsetlinewidth{0.060000\du}
\pgfsetdash{}{0pt}
\pgfsetdash{}{0pt}
\pgfsetbuttcap
\pgfsetmiterjoin
\pgfsetlinewidth{0.060000\du}
\pgfsetbuttcap
\pgfsetmiterjoin
\pgfsetdash{}{0pt}
\definecolor{dialinecolor}{rgb}{0.529412, 0.588235, 0.772549}
\pgfsetfillcolor{dialinecolor}
\pgfpathellipse{\pgfpoint{24.981457\du}{7.605134\du}}{\pgfpoint{0.360000\du}{0\du}}{\pgfpoint{0\du}{0.360000\du}}
\pgfusepath{fill}
\definecolor{dialinecolor}{rgb}{0.000000, 0.000000, 0.000000}
\pgfsetstrokecolor{dialinecolor}
\pgfpathellipse{\pgfpoint{24.981457\du}{7.605134\du}}{\pgfpoint{0.360000\du}{0\du}}{\pgfpoint{0\du}{0.360000\du}}
\pgfusepath{stroke}
\pgfsetbuttcap
\pgfsetmiterjoin
\pgfsetdash{}{0pt}
\definecolor{dialinecolor}{rgb}{0.000000, 0.000000, 0.000000}
\pgfsetstrokecolor{dialinecolor}
\pgfpathellipse{\pgfpoint{24.981457\du}{7.605134\du}}{\pgfpoint{0.360000\du}{0\du}}{\pgfpoint{0\du}{0.360000\du}}
\pgfusepath{stroke}
\pgfsetlinewidth{0.050000\du}
\pgfsetdash{}{0pt}
\pgfsetdash{}{0pt}
\pgfsetbuttcap
{
\definecolor{dialinecolor}{rgb}{0.000000, 0.000000, 0.000000}
\pgfsetfillcolor{dialinecolor}
\definecolor{dialinecolor}{rgb}{0.000000, 0.000000, 0.000000}
\pgfsetstrokecolor{dialinecolor}
\draw (21.145643\du,10.383379\du)--(24.669570\du,7.831032\du);
}
\pgfsetlinewidth{0.050000\du}
\pgfsetdash{}{0pt}
\pgfsetdash{}{0pt}
\pgfsetbuttcap
{
\definecolor{dialinecolor}{rgb}{0.000000, 0.000000, 0.000000}
\pgfsetfillcolor{dialinecolor}
\definecolor{dialinecolor}{rgb}{0.000000, 0.000000, 0.000000}
\pgfsetstrokecolor{dialinecolor}
\draw (24.863164\du,7.976879\du)--(24.024838\du,10.611385\du);
}
\pgfsetlinewidth{0.050000\du}
\pgfsetdash{}{0pt}
\pgfsetdash{}{0pt}
\pgfsetbuttcap
{
\definecolor{dialinecolor}{rgb}{0.000000, 0.000000, 0.000000}
\pgfsetfillcolor{dialinecolor}
\definecolor{dialinecolor}{rgb}{0.000000, 0.000000, 0.000000}
\pgfsetstrokecolor{dialinecolor}
\draw (25.167686\du,7.946914\du)--(26.432125\du,10.267496\du);
}
\pgfsetlinewidth{0.070000\du}
\pgfsetdash{}{0pt}
\pgfsetdash{}{0pt}
\pgfsetbuttcap
\pgfsetmiterjoin
\pgfsetlinewidth{0.070000\du}
\pgfsetbuttcap
\pgfsetmiterjoin
\pgfsetdash{}{0pt}
\definecolor{dialinecolor}{rgb}{1.000000, 1.000000, 1.000000}
\pgfsetfillcolor{dialinecolor}
\fill (20.827220\du,10.611385\du)--(22.250000\du,14.853718\du)--(19.404441\du,14.853718\du)--cycle;
\definecolor{dialinecolor}{rgb}{0.000000, 0.000000, 0.000000}
\pgfsetstrokecolor{dialinecolor}
\draw (20.827220\du,10.611385\du)--(22.250000\du,14.853718\du)--(19.404441\du,14.853718\du)--cycle;
\definecolor{dialinecolor}{rgb}{0.000000, 0.000000, 0.000000}
\pgfsetstrokecolor{dialinecolor}
\node[anchor=west] at (20.061882\du,14.433628\du){$B_k$};
\pgfsetlinewidth{0.050000\du}
\pgfsetdash{}{0pt}
\pgfsetdash{}{0pt}
\pgfsetbuttcap
{
\definecolor{dialinecolor}{rgb}{0.000000, 0.000000, 0.000000}
\pgfsetfillcolor{dialinecolor}
\definecolor{dialinecolor}{rgb}{0.000000, 0.000000, 0.000000}
\pgfsetstrokecolor{dialinecolor}
\draw (25.317638\du,7.803345\du)--(29.740540\du,10.411066\du);
}
\pgfsetlinewidth{0.070000\du}
\pgfsetdash{}{0pt}
\pgfsetdash{}{0pt}
\pgfsetbuttcap
\pgfsetmiterjoin
\pgfsetlinewidth{0.070000\du}
\pgfsetbuttcap
\pgfsetmiterjoin
\pgfsetdash{}{0pt}
\definecolor{dialinecolor}{rgb}{1.000000, 1.000000, 1.000000}
\pgfsetfillcolor{dialinecolor}
\fill (24.024838\du,10.611385\du)--(25.163062\du,14.005252\du)--(22.886614\du,14.005252\du)--cycle;
\definecolor{dialinecolor}{rgb}{0.000000, 0.000000, 0.000000}
\pgfsetstrokecolor{dialinecolor}
\draw (24.024838\du,10.611385\du)--(25.163062\du,14.005252\du)--(22.886614\du,14.005252\du)--cycle;
\pgfsetlinewidth{0.075600\du}
\pgfsetdash{}{0pt}
\pgfsetdash{}{0pt}
\pgfsetbuttcap
\pgfsetmiterjoin
\pgfsetlinewidth{0.075600\du}
\pgfsetbuttcap
\pgfsetmiterjoin
\pgfsetdash{}{0pt}
\definecolor{dialinecolor}{rgb}{1.000000, 1.000000, 1.000000}
\pgfsetfillcolor{dialinecolor}
\fill (26.620247\du,10.602185\du)--(27.603672\du,13.534486\du)--(25.636821\du,13.534486\du)--cycle;
\definecolor{dialinecolor}{rgb}{0.000000, 0.000000, 0.000000}
\pgfsetstrokecolor{dialinecolor}
\draw (26.620247\du,10.602185\du)--(27.603672\du,13.534486\du)--(25.636821\du,13.534486\du)--cycle;
\pgfsetlinewidth{0.070000\du}
\pgfsetdash{}{0pt}
\pgfsetdash{}{0pt}
\pgfsetbuttcap
\pgfsetmiterjoin
\pgfsetlinewidth{0.070000\du}
\pgfsetbuttcap
\pgfsetmiterjoin
\pgfsetdash{}{0pt}
\definecolor{dialinecolor}{rgb}{1.000000, 1.000000, 1.000000}
\pgfsetfillcolor{dialinecolor}
\fill (30.072908\du,10.611385\du)--(30.801371\du,12.783459\du)--(29.344445\du,12.783459\du)--cycle;
\definecolor{dialinecolor}{rgb}{0.000000, 0.000000, 0.000000}
\pgfsetstrokecolor{dialinecolor}
\draw (30.072908\du,10.611385\du)--(30.801371\du,12.783459\du)--(29.344445\du,12.783459\du)--cycle;
\definecolor{dialinecolor}{rgb}{0.000000, 0.000000, 0.000000}
\pgfsetstrokecolor{dialinecolor}
\node[anchor=west] at (22.938469\du,13.589307\du){$B_{k-1}$};
\pgfsetlinewidth{0.060000\du}
\pgfsetdash{}{0pt}
\pgfsetdash{}{0pt}
\pgfsetbuttcap
\pgfsetmiterjoin
\pgfsetlinewidth{0.060000\du}
\pgfsetbuttcap
\pgfsetmiterjoin
\pgfsetdash{}{0pt}
\definecolor{dialinecolor}{rgb}{0.529412, 0.588235, 0.772549}
\pgfsetfillcolor{dialinecolor}
\pgfpathellipse{\pgfpoint{24.016126\du}{10.609276\du}}{\pgfpoint{0.360000\du}{0\du}}{\pgfpoint{0\du}{0.360000\du}}
\pgfusepath{fill}
\definecolor{dialinecolor}{rgb}{0.000000, 0.000000, 0.000000}
\pgfsetstrokecolor{dialinecolor}
\pgfpathellipse{\pgfpoint{24.016126\du}{10.609276\du}}{\pgfpoint{0.360000\du}{0\du}}{\pgfpoint{0\du}{0.360000\du}}
\pgfusepath{stroke}
\pgfsetbuttcap
\pgfsetmiterjoin
\pgfsetdash{}{0pt}
\definecolor{dialinecolor}{rgb}{0.000000, 0.000000, 0.000000}
\pgfsetstrokecolor{dialinecolor}
\pgfpathellipse{\pgfpoint{24.016126\du}{10.609276\du}}{\pgfpoint{0.360000\du}{0\du}}{\pgfpoint{0\du}{0.360000\du}}
\pgfusepath{stroke}
\definecolor{dialinecolor}{rgb}{0.000000, 0.000000, 0.000000}
\pgfsetstrokecolor{dialinecolor}
\node[anchor=west] at (29.325916\du,12.344314\du){$B_0$};
\pgfsetlinewidth{0.060000\du}
\pgfsetdash{}{0pt}
\pgfsetdash{}{0pt}
\pgfsetbuttcap
\pgfsetmiterjoin
\pgfsetlinewidth{0.060000\du}
\pgfsetbuttcap
\pgfsetmiterjoin
\pgfsetdash{}{0pt}
\definecolor{dialinecolor}{rgb}{0.529412, 0.588235, 0.772549}
\pgfsetfillcolor{dialinecolor}
\pgfpathellipse{\pgfpoint{26.618354\du}{10.609276\du}}{\pgfpoint{0.360000\du}{0\du}}{\pgfpoint{0\du}{0.360000\du}}
\pgfusepath{fill}
\definecolor{dialinecolor}{rgb}{0.000000, 0.000000, 0.000000}
\pgfsetstrokecolor{dialinecolor}
\pgfpathellipse{\pgfpoint{26.618354\du}{10.609276\du}}{\pgfpoint{0.360000\du}{0\du}}{\pgfpoint{0\du}{0.360000\du}}
\pgfusepath{stroke}
\pgfsetbuttcap
\pgfsetmiterjoin
\pgfsetdash{}{0pt}
\definecolor{dialinecolor}{rgb}{0.000000, 0.000000, 0.000000}
\pgfsetstrokecolor{dialinecolor}
\pgfpathellipse{\pgfpoint{26.618354\du}{10.609276\du}}{\pgfpoint{0.360000\du}{0\du}}{\pgfpoint{0\du}{0.360000\du}}
\pgfusepath{stroke}
\definecolor{dialinecolor}{rgb}{0.000000, 0.000000, 0.000000}
\pgfsetstrokecolor{dialinecolor}
\node[anchor=west] at (25.539866\du,13.127369\du){$B_{k-2}$};
\pgfsetlinewidth{0.060000\du}
\pgfsetdash{}{0pt}
\pgfsetdash{}{0pt}
\pgfsetbuttcap
\pgfsetmiterjoin
\pgfsetlinewidth{0.060000\du}
\pgfsetbuttcap
\pgfsetmiterjoin
\pgfsetdash{}{0pt}
\definecolor{dialinecolor}{rgb}{0.529412, 0.588235, 0.772549}
\pgfsetfillcolor{dialinecolor}
\pgfpathellipse{\pgfpoint{20.833756\du}{10.609276\du}}{\pgfpoint{0.360000\du}{0\du}}{\pgfpoint{0\du}{0.360000\du}}
\pgfusepath{fill}
\definecolor{dialinecolor}{rgb}{0.000000, 0.000000, 0.000000}
\pgfsetstrokecolor{dialinecolor}
\pgfpathellipse{\pgfpoint{20.833756\du}{10.609276\du}}{\pgfpoint{0.360000\du}{0\du}}{\pgfpoint{0\du}{0.360000\du}}
\pgfusepath{stroke}
\pgfsetbuttcap
\pgfsetmiterjoin
\pgfsetdash{}{0pt}
\definecolor{dialinecolor}{rgb}{0.000000, 0.000000, 0.000000}
\pgfsetstrokecolor{dialinecolor}
\pgfpathellipse{\pgfpoint{20.833756\du}{10.609276\du}}{\pgfpoint{0.360000\du}{0\du}}{\pgfpoint{0\du}{0.360000\du}}
\pgfusepath{stroke}
\pgfsetlinewidth{0.060000\du}
\pgfsetdash{}{0pt}
\pgfsetdash{}{0pt}
\pgfsetbuttcap
\pgfsetmiterjoin
\pgfsetlinewidth{0.060000\du}
\pgfsetbuttcap
\pgfsetmiterjoin
\pgfsetdash{}{0pt}
\definecolor{dialinecolor}{rgb}{0.529412, 0.588235, 0.772549}
\pgfsetfillcolor{dialinecolor}
\pgfpathellipse{\pgfpoint{30.076721\du}{10.609276\du}}{\pgfpoint{0.360000\du}{0\du}}{\pgfpoint{0\du}{0.360000\du}}
\pgfusepath{fill}
\definecolor{dialinecolor}{rgb}{0.000000, 0.000000, 0.000000}
\pgfsetstrokecolor{dialinecolor}
\pgfpathellipse{\pgfpoint{30.076721\du}{10.609276\du}}{\pgfpoint{0.360000\du}{0\du}}{\pgfpoint{0\du}{0.360000\du}}
\pgfusepath{stroke}
\pgfsetbuttcap
\pgfsetmiterjoin
\pgfsetdash{}{0pt}
\definecolor{dialinecolor}{rgb}{0.000000, 0.000000, 0.000000}
\pgfsetstrokecolor{dialinecolor}
\pgfpathellipse{\pgfpoint{30.076721\du}{10.609276\du}}{\pgfpoint{0.360000\du}{0\du}}{\pgfpoint{0\du}{0.360000\du}}
\pgfusepath{stroke}
\pgfsetlinewidth{0.008000\du}
\pgfsetdash{}{0pt}
\pgfsetdash{}{0pt}
\pgfsetbuttcap
\pgfsetmiterjoin
\pgfsetlinewidth{0.008000\du}
\pgfsetbuttcap
\pgfsetmiterjoin
\pgfsetdash{}{0pt}
\definecolor{dialinecolor}{rgb}{0.000000, 0.000000, 0.000000}
\pgfsetfillcolor{dialinecolor}
\pgfpathellipse{\pgfpoint{28.999404\du}{10.612546\du}}{\pgfpoint{0.089013\du}{0\du}}{\pgfpoint{0\du}{0.089013\du}}
\pgfusepath{fill}
\definecolor{dialinecolor}{rgb}{0.000000, 0.000000, 0.000000}
\pgfsetstrokecolor{dialinecolor}
\pgfpathellipse{\pgfpoint{28.999404\du}{10.612546\du}}{\pgfpoint{0.089013\du}{0\du}}{\pgfpoint{0\du}{0.089013\du}}
\pgfusepath{stroke}
\pgfsetbuttcap
\pgfsetmiterjoin
\pgfsetdash{}{0pt}
\definecolor{dialinecolor}{rgb}{0.000000, 0.000000, 0.000000}
\pgfsetstrokecolor{dialinecolor}
\pgfpathellipse{\pgfpoint{28.999404\du}{10.612546\du}}{\pgfpoint{0.089013\du}{0\du}}{\pgfpoint{0\du}{0.089013\du}}
\pgfusepath{stroke}
\pgfsetlinewidth{0.008000\du}
\pgfsetdash{}{0pt}
\pgfsetdash{}{0pt}
\pgfsetbuttcap
\pgfsetmiterjoin
\pgfsetlinewidth{0.008000\du}
\pgfsetbuttcap
\pgfsetmiterjoin
\pgfsetdash{}{0pt}
\definecolor{dialinecolor}{rgb}{0.000000, 0.000000, 0.000000}
\pgfsetfillcolor{dialinecolor}
\pgfpathellipse{\pgfpoint{27.775793\du}{10.609697\du}}{\pgfpoint{0.086164\du}{0\du}}{\pgfpoint{0\du}{0.086164\du}}
\pgfusepath{fill}
\definecolor{dialinecolor}{rgb}{0.000000, 0.000000, 0.000000}
\pgfsetstrokecolor{dialinecolor}
\pgfpathellipse{\pgfpoint{27.775793\du}{10.609697\du}}{\pgfpoint{0.086164\du}{0\du}}{\pgfpoint{0\du}{0.086164\du}}
\pgfusepath{stroke}
\pgfsetbuttcap
\pgfsetmiterjoin
\pgfsetdash{}{0pt}
\definecolor{dialinecolor}{rgb}{0.000000, 0.000000, 0.000000}
\pgfsetstrokecolor{dialinecolor}
\pgfpathellipse{\pgfpoint{27.775793\du}{10.609697\du}}{\pgfpoint{0.086164\du}{0\du}}{\pgfpoint{0\du}{0.086164\du}}
\pgfusepath{stroke}
\pgfsetlinewidth{0.008000\du}
\pgfsetdash{}{0pt}
\pgfsetdash{}{0pt}
\pgfsetbuttcap
\pgfsetmiterjoin
\pgfsetlinewidth{0.008000\du}
\pgfsetbuttcap
\pgfsetmiterjoin
\pgfsetdash{}{0pt}
\definecolor{dialinecolor}{rgb}{0.000000, 0.000000, 0.000000}
\pgfsetfillcolor{dialinecolor}
\pgfpathellipse{\pgfpoint{28.373847\du}{10.608371\du}}{\pgfpoint{0.084838\du}{0\du}}{\pgfpoint{0\du}{0.084838\du}}
\pgfusepath{fill}
\definecolor{dialinecolor}{rgb}{0.000000, 0.000000, 0.000000}
\pgfsetstrokecolor{dialinecolor}
\pgfpathellipse{\pgfpoint{28.373847\du}{10.608371\du}}{\pgfpoint{0.084838\du}{0\du}}{\pgfpoint{0\du}{0.084838\du}}
\pgfusepath{stroke}
\pgfsetbuttcap
\pgfsetmiterjoin
\pgfsetdash{}{0pt}
\definecolor{dialinecolor}{rgb}{0.000000, 0.000000, 0.000000}
\pgfsetstrokecolor{dialinecolor}
\pgfpathellipse{\pgfpoint{28.373847\du}{10.608371\du}}{\pgfpoint{0.084838\du}{0\du}}{\pgfpoint{0\du}{0.084838\du}}
\pgfusepath{stroke}
\definecolor{dialinecolor}{rgb}{0.000000, 0.000000, 0.000000}
\pgfsetstrokecolor{dialinecolor}
\node[anchor=west] at (26.180333\du,14.772055\du){$B_k$};
\end{tikzpicture}
\end{minipage}
\caption{An example of a binomial heap of size $2^{k+1}$.}
\end{figure}

Let us call the entire structure a \textit{priority queue}/\textit{Binomial
heap} and a particular tree of the size of $2^k$ for some $k$ a \textit{heap}.

Similarly to the \textit{Binary Heapsort}, the sorting procedure consists of two
phases: building the \textit{priority queue} and popping the maximum element out
of the \textit{priority queue} until it is empty.

At the heart of the algorithm is the merge procedure, which takes two \textit{
Root list}s and merges them into another \textit{Root list}. The procedure 
is very similar to binary addition of numbers. 

An auxiliary $Add$ procedure to deal with certain special cases of merging
described below:
\algrenewcommand\Return{\State \algorithmicreturn{} }%

\begin{algorithmic}[1]
\Function{Add}{$L$,$node$,$carry$}
\Comment{Merges the heaps $node$ and the $carry$ in case both are of the same
size. When both are of different size, adds them to the list $L$}
\If{$carry = null$}
  \State{$L.add(node)$}
  \Return{$0$}
\Else
  \If{$carry.size = node.size$}
    \Return{$Merge(node, carry)$}
  \Else
    \State{$L.add(node)$}
    \State{$L.add(carry)$}
    \Return $null$
  \EndIf
\EndIf
\EndFunction
\end{algorithmic}

A $Merge$ procedure to merge two heaps of the same size:

\begin{algorithmic}[1]
\Function{Merge}{$L$,$node1$,$node2$}
\Comment{Merges the heaps $node1$ and $node2$ into a heap of twice the size}
\If{$node1.value > node2.value$}
  \State{$node1.children.add(node2)$}
  \Return{$node1$}
\Else
  \State{$node2.children.add(node1)$}
  \Return{$node2$}
\EndIf
\EndFunction
\end{algorithmic}

Finally another $Merge$ procedure for merging two \textit{Root lists} lists
together:

\begin{algorithmic}[1]
\Function{Merge}{$L_1$,$L_2$}
\Comment{Merges the root lists $L_1$ and $L_2$ to produce the list $L_3$
size}
\label{alg:merge}
\State{$L_3 \gets \{\}$}
\State{$carry \gets null$}
\State{$t_1 \gets L_1[0], \: t_2 \gets L_2[0]$}
\While{$t_1 \neq null \: {\bf and} \: t_2 \neq null$}
  \If{$t_1.size = t_2.size$}
     \If{$carry \neq null$}
       \State{$L_3.add(carry)$}
     \EndIf
     \State{$carry \gets Merge(t_1, t_2)$}
     \State{$t_1 \gets next(t_1), \: t_2 \gets next(t_2)$}
  \Else
     \If{$carry \neq null$}
       \If{$t_1.size < t_2.size$}
         \State{$carry \gets Add(L_3, carry, t_1)$}
         \State{$t_1 \gets next(t_1)$}
       \Else
         \State{$carry \gets Add(L_3, carry, t_2)$}
         \State{$t_2 \gets next(t_2)$}
       \EndIf
     \EndIf
  \EndIf
\EndWhile

\While{$t_1 \neq null$}
  \State{$carry \gets Add(L_3, carry, t_1)$}
  \State{$t_1 \gets next(t_1)$}
\EndWhile

\While{$t_2 \neq null$}
  \State{$carry \gets Add(L_3, carry, t_2)$}
  \State{$t_2 \gets next(t_2)$}
\EndWhile
\Return{$L_3$}
\EndFunction
\end{algorithmic}

When inserting an element $x$ to the \textit{priority queue}, we create a heap
of size $1$ containing just the element $x$ and merge it with the \textit{Root
list}. 

To remove the maximal element from the \textit{priority queue}, we firstly
need to find it in the \textit{Root list}. Afterward, we merge every child of
the maximal element back to the \textit{Root list}. The details of all the
procedures can be found in ~\cite{Vuillemin:1978:DSM:359460.359478}.

\section{Randomness preservation}

The first thing we would like to show is that the \textit{Binomial Heapsort}
preserves randomness of the \textit{priority queue} during the execution. Just
as was done for the \textit{Heap sort}, suppose that the input array is shuffled
and that the distribution of permutations of the input elements is uniformly
random.

Let us investigate the probability of a particular \textit{priority queue}
occurring.

\begin{myle}
Let $n/2 < 2^k \leq n$, $B_n = {n \choose 2^k} B_{2^k}B_{n-2^k}$ for $n > 1$
and $B_1 = 1$.
Then the number of possible \textit{priority queue}s on $n$ elements is $B_n$.
\end{myle}
\begin{proof}
The statement of the lemma follows from the fact that the resulting \textit{priority queue} is
uniformly random, that is any subset of elements is equally likely to appear in 
a heap that has its root in the \textit{Root list}. In particular, the largest
heap containing $2^k$ elements is also uniformly random.
\end{proof}

\begin{myde}
A uniformly random \textit{Binomial heap} of size $n$ is such that each
configuration of the elements obeying the heap order is equally likely, that is,
has probability to occur of $1/B_n$.
\end{myde}

The problem of randomness preservation of classical \textit{Heapsort} was
indicated by the fact that the number $H_{n}$ (the number of possible
\textit{Binary heap}s on $n$ elements) was not divisible by $H_{n-1}$ . Hence
after a deletion of an element from a uniformly random heap of $n$ elements, the
heap could not possibly stay uniformly random.

However with \textit{Binomial heap}s there is a hope to show randomness
preservation:

\begin{myle}
Let $B_n$ be defined as above, then it holds that $B_n/B_{n-1} = n$ for
odd $n$ and $B_n/B_{n-1} = n/2$ for even $n$.
\end{myle}
\begin{proof}
We can show this fact by induction. There are $2$ cases to consider: 
$n=2^{m}$ and $n \neq 2^m$ for any $m$. In the first case $B_n/B_{n-1}$ = $n
\choose n/2$ $B_{n/2}B_{n/2}/$ $n-1 \choose n/2$ $B_{n-1-n/2}B_{n/2}$ =
$n(n-1-n/2)B_{n/2}/(n/2)B_{n-1-n/2} = n/2$ by induction hypothesis. With a
similar reasoning one can demonstrate that the second case holds. The base case
of induction holds clearly as $B_2 = 1$, $B_1 = 1$.
\end{proof}

Using two lemmas above, show that after the \textit{BuildHeap} procedure we
obtain a uniformly random \textit{heap}. Suppose there are two uniformly random
\textit{heap}s $T_1, T_2$ of the same size. Fix the set of the elements
contained in $T_1 \cap T_2$ to be $S$.

\begin{myle}
After merging $T_1$ and $T_2$, the resulting \textit{heap} $T$ is uniformly random on the 
set $S$.
\end{myle}
\begin{proof}
Probability of a particular \textit{heap} $\textit{join}(T_1, T_2)$ is ${ S \choose
|T_1|}{S-|T_1| \choose T_2}B_{|T_1|}{B_{|T_2|}}$ by construction, which is the
same as probability of a particular $2$-tuple $(T_1, T_2)$. As the set $S$ was chosen
arbitrarily, the argument works for any set $S$ of an appropriate size.
\end{proof}

\begin{myle}
After the \textit{BuildHeap} procedure, the \textit{Binomial Heap} is uniformly
random, in a sense that, any \textit{Binomial heap} configuration of the input elements is
equally likely.
\end{myle}

One of the main lemmas of the section follows immediately:

\begin{proof}
We can show the fact inductively, adding input elements one by one. Clearly,
when we add the first element, the heap is uniformly random, as the first
element in the input array is uniformly random. Let us suppose that we have
added $p$ elements, and we are adding $p+1$st and so far any \textit{Binomial
heap} configuration on $p$ elements is equally likely. 

Let the \textit{priority queue} on $p$ elements be $T_p$, then the probability of a particular
tuple $(x_{p+1}, T_p)$ is the same as the probability $\textit{merge}(T_p,
\{x_{p+1}\})$ in case the \textit{Root list} of $T_p$ does not contain a heap of
size $1$. In case it does contain, by the lemma above, \textit{join} of two
heaps having size a power of $2$ is also uniformly random. Now we have
an ``overflow" heap of size $2$. Clearly, we can extend the argument in case
there is a heap of size $2, 4, \ldots$ as well. Hence, after merging an element
$x_{p+1}$, the resulting \textit{priority queue} is uniformly random.
\end{proof}

\begin{myco}
Every fixed subset of heaps is also uniformly random, in a sense that, every
possible distribution of elements within the heaps is equally likely. By
construction, every subheap of a heap of size $2^s$ is also uniformly random.
\end{myco}

So far we have shown that when the \textit{priority queue} is built from a uniformly permuted
array, its distribution is also uniformly random. What is left to show is that
after a deletion of a maximum element in the \textit{Root list}, the \textit{priority queue} stays
uniformly random.

\begin{myle}
After deletion of a maximal element in a uniformly random \textit{Binomial
heap}, the resulting \textit{priority queue} is also uniformly random.
\end{myle}
\begin{proof}
The proof of this lemma is slightly more involved than the previous one, as
the position of maximal element in a uniformly random \textit{priority queue}'s
\textit{Root list} is not deterministic. Here is an outline of the proof: we will condition on the
event that the maximal element in the \textit{Root list} is in a particular
heap. Then we will show that conditioned on this event, the resulting \textit{priority queue} is
uniformly random. In the end we will use the law of total probability to compute 
the probability of a particular \textit{priority queue} occurring.

\begin{myob}
Define the set $B_n$ as the set of \textit{Binomial heaps} with $n$ elements.
Let the set $B^p_n$ be the set of all the \textit{priority queue}s of size $n$
that have a maximal element in the heap of size $2^p$. Then the map $PopMax:B^p_n \rightarrow
B_{n-1}$ is surjective in a sense that for every heap in $B_{n-1}$ there exists
a heap in $B_{n+1}^p$ such that a $PopMax$ transforms the later heap into the
former. Define the event $E^p_n$ indicating that the maximal element is in the
heap of size $2^p$ for a \textit{priority queue} of size $n$.
\end{myob}

\begin{myob}
For all the \textit{priority queue}s $B \in B_n$, the sizes of the preimages
under the map $PopMax$ are equal.
\end{myob}

The observation can be shown to be true by noticing that 
\textit{Binomial heap}s are isomorphic under relabelling. The two observations
above indicate that the probability of a particular \textit{Binomial heap} $B$
after the \textit{PopMax} operation is independent of $B$, and hence any $B$ has
the same probability ($1/B_{n-1}$) of occurring.

To finish the proof, the events $E^p_n$ for a fixed $n$ and different $p$ form a
disjoint partition of the probability space, and hence we can apply the law of
the total probability. The event $E_B$ indicating a particular \textit{priority
queue} occurring after a \textit{PopMax} operation, has the probability
\begin{equation}
Pr_{E_B} = \sum_p Pr(E_B | E^p_n)P(E^p_n) = \sum_p P(E^p_n)/B_{n-1} = 1/B_{n-1}
\end{equation}
as
\begin{equation}
\sum_p P(E^p_n) = 1
\end{equation}
\end{proof}

\section{Number of comparisons}

Once again, we would like to count how many comparisons does the sorting
algorithm cause between a contiguous range of elements of size $r$. Let us argue
that the number of comparisons during \textit{BuildHeap} phase is linear in
terms of the $r$.

\begin{myle}(Bounding number of comparisons during the \textit{BuildHeap}
phase) The number of comparisons between the elements of a range of size $r$ during the
\textit{BuildHeap} phase is at most $r$.
\end{myle}
\begin{proof}
We just have to notice that when two \textcolor{red}{red} elements
are compared, their trees are merged and only the larger of the roots is 
``available'' for further comparisons. Hence with every comparison of
\textcolor{red}{red} elements, we have $1$ less elements which can cause
further comparisons.
\end{proof}

Consider a moment, when the first \textcolor{red}{red} element is being popped.
We argue that before this moment, there are at most $r$ comparisons between the
\textcolor{red}{red} elements.

\begin{myle}
If two \textcolor{red}{red} elements are compared, one of them stays the
ancestor of the other until the moment they are unmerged (the ancestor is
deleted).
This can only happen when a \textcolor{red}{red} element is popped
from the \textit{priority queue}.
\end{myle}

Hence we only need to consider the situation when the range $X_r$ is the $r$
largest elements of in the \textit{priority queue}, as up to this event, there can be only $r$ comparisons
between the \textcolor{red}{red} elements. 

So from now on we assume that there are $n$
elements in the \textit{priority queue} and our range $X_r$ is the $r$ largest elements left. 
Note that this is where we need randomness preservation property, as we can
guarantee that at any point in the algorithm, the \textit{priority queue} has
uniform distribution.

\begin{mythe}
The number of comparisons between the \textcolor{red}{red} elements,
caused by searching the maximum element in the \textit{priority queue} is at
most $2\ln(2)r\log(r)$.
\end{mythe}

To establish the result above, let us look at the probabilities of particular
\textit{priority queue} configurations which cause comparisons between the
\textcolor{red}{red} elements.

Clearly, if there are $k$ \textcolor{red}{red} elements in the \textit{Root
list}, there would be at most $k-1$ comparisons between the
\textcolor{red}{red} elements when searching the maximum element of
the \textit{Root list}.

The probability that $k$-th smallest \textcolor{red}{red} element is a root of
a heap of size $2^{s}$ is ${ k-1 \choose 2^{s}-1}/{n \choose 2^s}$. Summing the
probability for all \textcolor{red}{red} elements, we have 
\begin{equation}
\sum_{n-r+1 \leq k \leq n} { k-1 \choose
2^{s}-1}/{n \choose 2^s} = 1-{n-r \choose 2^s}/{n \choose 2^s}
\end{equation}

\begin{myle}
The probability that any \textcolor{red}{red} element is the root
of a heap of size $2^s$ is 
\begin{equation}
1-{n-r \choose 2^s}/{n \choose 2^s}
\end{equation}
\end{myle}

A direct result of the lemma is that the expected number of comparisons while
searching for the maximal element:

\begin{myle}
Let the binary expansion of $n$ be $\displaystyle \sum_{s_i} 2^{s_i}$. Then the
expected number of \textcolor{red}{red} elements that are in the
\textit{Root list} is
\begin{equation}
\sum_{s_i} 1-{n-r \choose 2^{s_i}}/{n \choose 2^{s_i}}
\end{equation}

such that $2^{s_i}$ occurs in the binary expansion of $n$.
\end{myle}

Notice that, after each $Pop$ operation, both $n$ and $r$ decrease by $1$, as
the maximal element in the \textit{Root list} is \textcolor{red}{red} and is being
popped.
To get the expected overall number of comparisons during the phase, we can sum
up the expected numbers of comparisons while decreasing $n$. The only difficulty
is that the expression depends on the binary expansion of $n$.

Let us look at occurrence of $2^s$ in the binary expansion of $n$, while
decreasing n. It is not hard to realize that $2^s$ occurs in the binary
expansion of $n$ in blockwise fashion: let $n = 2^m-1$, then $2^s$ will occur in
the binary expansion of $n, n-1, \ldots n-r+1$, not occur in the binary
expansion of $n-r, n-r-1, \ldots n-2r+1$ and so on.

Let the initial number of elements in the \textit{priority queue} be $n_0$. For
the sake of simplifying the notation, denote $\oplus$ the \textit{xor} operation
and $n \oplus 2^s < n$ the case that $2^s$ is present in the binary expansion of
$n$.
Summing the quantity in lemma above for all $n$ such that $n_0-r+1 \leq n \leq
n_0$, as for smaller $n$, no \textcolor{red}{red} elements are left in the heap
(and hence no more comparisons are created):

\begin{equation}
\sum_{\substack{0 \leq k \leq r-1 \\ (n-k)
\oplus 2^s < n-k}} 1-{n-r \choose
2^{s}}/{n-k \choose 2^{s}}
\end{equation}

In turns out that analyzing the asymptotics of the expression is hard while
preserving the right constants.
One could upper bound the expression by the following expression:

\begin{equation}
\sum_{0 \leq k \leq r-1} 1-{n-r \choose
2^{s}}/{n-k \choose 2^{s}}
\end{equation}

However experiments show that this approximation is approximately
a $2$-competitive upper bound. Instead, we rely on other means to establish the
asymptotic behavior of the number of comparisons. 

What we really need to count is the expected number of heaps
that contain at least $1$ \textcolor{red}{red} element. We can see the heaps of sizes $2^s$ for
some $s$ as buckets where we put \textcolor{red}{red} elements.

\begin{myob}
At any moment, the number of elements in a heap of size $2^s$ is larger than the 
number of elements in all the smaller heaps combined.
\end{myob}
\begin{proof}
This is rather easy to see, as the number of elements in heaps is determined 
by the binary expansion of $n$.
\end{proof}

The next corollary should give intuition of why we would expect $O(\log(r))$
heaps to contain at least $1$ \textcolor{red}{red} element and
hence the number of comparisons between the \textcolor{red}{red} elements is
$O(\log(r))$ during finding the maximal \textcolor{red}{red} element.

\begin{myco}
If $n=2^k-1$, we expect about $1/2$ of the \textcolor{red}{red} elements to be
in the largest heap and the other $1/4$ to be in the next largest heap, and so on. 
For $n$ of a different form, even larger of fractions of elements  occur in larger
heaps. Hence we expect exponential decay in the number of elements ``left to
put'' in the smaller heaps.
\end{myco}

Consider the following process: number the heaps from largest to smallest
starting from $1$. We will put some number of the \textcolor{red}{red} elements
to the heaps in this order. Let $X_t$ be the number of
\textcolor{red}{red} elements left (of initial size $r$) before we
put a number of \textcolor{red}{red} elements to the $t$-th heap.

\begin{myde}
Let $T = \min \{t \in \mathbb{N}_0 | X_t = 0 \}$
\end{myde}

Clearly, the expected number of heaps that contain \textcolor{red}{red} elements
is $\leq T$, as there could be a heap that contains no \textcolor{red}{red}
elements between the heaps that have at least $1$ \textcolor{red}{red} element.

The following theorem, also known as a multiplicative drift lemma, establishes behavior of $T$
~\cite{DBLP:journals/algorithmica/DoerrJW12}:

\begin{mythe}(Multiplicative Drift lemma)
Let $\{X_t\}$ be a sequence of non-negative integer random variables. Assume
that there is a $\delta > 0$ such that
\begin{equation}
\forall t \in \mathbb{N}_0 \: : \: E(X_t | X_{t-1} = x) \leq (1-\delta)x
\end{equation}
then $T = \min \{t \in \mathbb{N}_0 | X_t = 0 \}$ satisfies
\begin{equation}
E(T) \leq (1/\delta)(\ln(X_0)+1)
\end{equation}
\end{mythe}

Let us show that our process with $\delta = 1/2$ satisfies the requirements of
the theorem:

\begin{proof}
As was shown before, the distribution of elements in the \textit{Binomial Heap}
is uniform, that is every subset of elements is equally likely to appear in a heap
of size $2^s$. Or more generally, for a collection of heaps of appropriate sizes
$s_1, s_2, \ldots s_k$, any set tuple $(S_1, S_2, \ldots, S_k)$ (where $|S_i| = s_i$) is
equally likely to appear as the corresponding sets of the heaps.

Now, for a single element, consider the probability $p$ that this element occurs
in the heap of size $2^s$ (which is the largest heap) and not in the heaps of
smaller size. It is not hard to see that $p > 1/2$, as the number of elements
in the heap of size $2^s$ is always larger than the combined number of elements in all the smaller
heaps (as $2^s > \sum_{0 \leq i \leq s-1}2^i$).

Let $H_t$ be the expected number of \textcolor{red}{red} elements
in $t$-th heap. Then $X_{t-1}-H_{t-1} = X_t$ and 
\begin{eqnarray}
E(X_t | X_{t-1} = x) & = & E(X_{t-1}-H_{t-1} |  X_{t-1} = x) = \\
x-E(H_{t-1}|X_{t-1} = x) & \leq & x-x/2 = x/2
\end{eqnarray}
As $E(H_{t-1}|X_{t-1} = x) \geq x p = x/2$. Hence the random process indeed
satisfies the requirements of the theorem and
\begin{equation}
E(T) \leq (1/\delta)(\ln(X_0)+1) = 2(\ln(r)+1)
\end{equation}
\end{proof}

The expectation of the total number of comparisons during the \textit{SearchMax}
phase is then $\leq 2\ln(r!)+2r = 2r\ln(r)+O(r)$. From the proof of the
Multiplicative Drift theorem it seems that the constant is rather tight. With a
more careful analysis, we can show that the number of elements in the
\textit{Root list} that we visit before we meet the maximal element is expected
to be $O(1)$.

Let the $p_i$ be the probability that the maximal element is in $i$-th heap,
then expected number of elements we visit in the \textit{Root list} 
before we meet the maximal element is:

\begin{equation}
\label{eq:binroot}
\sum_{1 \leq i \leq k} i p_{i}\prod_{1 \leq j < i}(1-p_j)
\end{equation}

The following inequalities help to establish the bound

\begin{equation}
q_1 q_2 \ldots q_k \leq \sqrt[k]{q_1 q_2 \ldots q_k} \leq (q_1+q_2+\ldots+q_k)/k
\end{equation}

for $0 \leq q_i \leq 1$. The first inequality follows from the fact that $q_i
\leq 1$ and the second is the standard arithmetic mean and geometric mean.

Applying the inequality to \ref{eq:binroot} leads to

\begin{equation}
\label{eq:binroot1}
\sum_{1 \leq i \leq k} i p_{i}\prod_{1 \leq j < i}(1-p_j) \leq \sum_{1 \leq i
\leq k} i p_{i}[\sum_{1 \leq j < i}(1-p_j)/(i-1)] = \sum_{1 \leq i \leq k} i
p_{i}
\end{equation}

The equality follows from the fact that $\sum_{1 \leq i \leq k}p_i = 1$. It is
easy to see that the maximum of the sum is attained when the sequence $(p_k,
p_{k-1}, \ldots, p_1)$ is lexicographically largest. However the construction of
the \textit{Binomial Heap} guarantees that $p_i \geq 2 p_{i+1}$ for all $j < k$,
as $i$-th heap is at least twice as large as the $i+1$-st. And hence, the sum 
$\sum_{1 \leq i \leq k}p_i = 1$ is at most $\sum_{1 \leq i \leq \infty}i/2^i =
2$. Note that this observation unfortunately doesn't help to find the maximal
element. 

Another possible improvement would be to put all the elements in the
\textit{Root list} to another \textit{Priority queue}. This does reduce the time
to find the maximal element to $O(\log(\log(n)))$ when the number of
\textcolor{red}{red} elements in is $O(n)$, however it is hard to argue how many
comparisons between the \textcolor{red}{red} elements this causes when $r$
(the number of \textcolor{red}{red} elements) is small, say $O(\log(n))$.

With the exact same technique, we can show that during the \textit{PopMax} phase
of the \textit{Binomial Heapsort} the number of comparisons of
\textcolor{red}{red} elements is also $O(r\log{r})$.
Once the maximal element in the \textit{Root list} has been found, the heap of
size $2^s$ is split into heaps of sizes $1, 2, \ldots, 2^{s-1}$, which are
remerged with the heaps that are currently in the \textit{Root list}. By the
randomness preservation argument, the distribution of the elements in the heaps
is still uniform.

\begin{myob}
When a comparison of two \textcolor{red}{red} elements happens, their
corresponding heaps are merged and there is $1$ less \textcolor{red}{red}
element in the \textit{Root list}.
\end{myob}

Intuitively then, the number of comparisons that happens when \textit{PopMax}
occurs is the difference of the number of \textcolor{red}{red} elements in the
\textit{Root lists} that are being merged (the list of children of
maximal element and the original \textit{Root list}).

\begin{myle}
Denote the expected number of \textcolor{red}{red} elements in the \textit{Root
list} of a \textit{priority queue} of size $n$ and having $r$
\textcolor{red}{red} elements overall as $R^r_n$. Then the expected number
of comparisons during a single \textit{PopMax} call is at most
\begin{equation}
2R^r_n-R^{r-1}_{n-1}
\end{equation}
\end{myle}
\begin{proof}
We have shown that $R^r_n \leq 2\ln(r)+1$. We can apply the very same
technique to show that the expected number of \textcolor{red}{red} elements that 
are immediate children of the maximal element (those that will be merged back
to the \textit{Root list}) is at most $R^r_n$. The statement of the lemma
reflects the previous observation. 
\end{proof}
Summing up the terms, we have 
\begin{equation}
\sum_{0 \leq k \leq r-1} 2R^{r-k}_{n-k}-R^{r-1-k}_{n-1-k} = R^r_n+\sum_{1
\leq k \leq r-1} R^{r-k}_{n-k} \leq 2\ln(r)+2\sum_{1 \leq k \leq r-1} \ln(k)
\end{equation}
which is at most $2\ln(r!) = 2r\ln(r)+O(r)$. It follows that during the
\textit{PopMax} phase there are at most $2r\ln(r)+O(r)$ comparisons of the
\textcolor{red}{red} elements, and hence the total number of comparisons
during the \textit{Binomial Heapsort} can be bounded by $4r\ln(r)+O(r)$. 
The experiments show that the bound is not tight and real
constant is around $2\ln(2)$, however it is not clear how to change the analysis
to reduce the constant, particularly during the \textit{FindMax} stage of the
algorithm.

\chapter{Experiments}\label{chap4}

\section{Binary heapsort}
Experiments support the analysis of the modified version of the \textit{Binary
Heapsort}. What is interesting is that even without the modifications proposed
in this thesis, based on experiments, it seems that the number of comparisons is
$r\log(r)+O(r)$, that is the right constant for the term $r\log(r)$ is $1$.

\section{Binomial heapsort}
For the \textit{Binomial heapsort}, the experiments suggest that the right
constant is $2\ln(2)$ instead of $4\ln(2)$, as was demonstrated in the analysis.
Experiments show that the algorithm spends only about half the predicted time in
both \textit{FindMax} and \textit{DeleteMax} phases, which suggests that the
constant derived in the Multiplicative Drift lemma
~\cite{DBLP:journals/algorithmica/DoerrJW12} is not tight. We can also
conclude from experiments, that the number of the \textcolor{red}{red} elements
in the \textit{Root list} is tightly concentrated around $\log(r)$,

\chapter{Conclusion}\label{chap5}

In this thesis we have conducted analysis of \textit{Binary Heapsort} algorithm
as when used to sort a collection of strings. Proposed modifications to the
algorithm provably guarantee that the total number of comparisons is small
(although not optimal in the asymptotic sense, see the dependence on the prefix
tree of the input strings). Unfortunately, modifications require additional
linear space, which leads to a non-in-place algorithm. However we believe that
the analysis presented is useful to understanding the nature of
\textit{Heapsort} algorithm.

For the \textit{Binomial Heapsort} we were not able to show a tight constant
factor for the number of comparisons of \textcolor{red}{red} elements, however
we were able to demonstrate that during execution of the algorithm, random
structure is preserved, which leads to a simpler and a more natural analysis
than in the case of \textit{Binary Heapsort}. 

Experiments conducted demonstrate that the constant factors in the
Multiplicative Drift lemma are not tight enough, which is an interesting 
research topic on its own. Without the use of the lemma however, the exact
expressions counting the number of comparisons are hard to analyze in the
asymptotic sense.


\addtocontents{toc}{\vspace{2em}} 



\appendix 


\chapter{Appendix}
\label{AppendixA}

\begin{myde}
(Rearrangement inequality) \\ Let $\{a_i | 1 \leq i \leq n \}$ be an increasing
sequence and $\{b_i|1 \leq i \leq n \} $ be decreasing. Then it holds that

\begin{center}
$\displaystyle \sum_i a_i b_i \leq \frac{(\sum_i a_i)(\sum_i b_i)}{n}$ 
\end{center}
\end{myde}

\begin{mythe}
(Euler-Maclaurin summation formula, first form) Let $f(x)$ be a function defined
on the interval $(1, \infty]$ and suppose that the derivatives $f^{(i)}(x)$
exist for $1 \leq i \leq 2m$, where $m$ is a
fixed constant. Then
\begin{center}
$\displaystyle \sum\limits_{1 \leq k \leq N}f(k) =
\int\limits_{1}^{N}f(x)\mathrm{d}x+\frac{(f(N)+f(1))}{2}+\sum\limits_{1 \leq k
\leq m}\frac{B_{2k}}{(2k)!}(f^{(2k-1)}(N)-f(1)^{(2k-1)})+R_m$
\end{center}
where $R_{m}$ is a
remainder s.t. $\displaystyle R_{m} =
-\int\limits_{1}^{N}B_{2m}\frac{f^{(2m)}(x)}{(2m)!}\{1-x\}\mathrm{d}x$ and $B_i$
is a Bernoulli number
\end{mythe}

\begin{mythe}
(Euler-Maclaurin summation formula, second form) Let $f(x)$ be a function defined
on the interval $(1, \infty]$ and suppose that the derivatives $f^{(i)}(x)$
exist and are absolutely integrable for $1 \leq i \leq 2m$, where $m$ is a
fixed constant. Then
\begin{center}
$\displaystyle \sum\limits_{1 \leq k \leq N}f(k) =
\int\limits_{1}^{N}f(x)\mathrm{d}x+\frac{1}{2}f(N)+C_f+\sum\limits_{1 \leq k
\leq m}\frac{B_{2k}}{(2k)!}f^{(2k-1)}(N)+R_m$
\end{center}
where $C_f$ is a constant associated with the function $f(x)$ and $R_{2m}$ is a
remainder term satisfying $\displaystyle |R_{2m}| =
O(\int\limits_{N}^{\infty}|f^{(2m)}(x)|\mathrm{d}x)$
\end{mythe}

\begin{myle}
(Asymptotics of  $\displaystyle \sum_{r > r' >
0}r'\log{r'}$)
\end{myle}
\begin{proof}

Fix $m=2$, then using Euler-Maclaurin formula (first form), we have
\begin{center}
$\displaystyle R_m =
\int\limits_{r-1}^{\infty}\frac{B_4}{4!}\frac{2\{1-x\}}{x^3}\mathrm{d}x
< \int\limits_{r-1}^{\infty}\frac{B_4}{4!}\frac{2}{x^3}\mathrm{d}x <
\frac{1}{(r-1)^2}$
\end{center}
Integrating $\displaystyle \int\limits_{1}^{r-1}x\log(x)\mathrm{d}x$ gives
\begin{center}
$\displaystyle
\frac{x^2}{2}\log(x)\bigg|_{1}^{r-1}-\int\limits_{1}^{r-1}\frac{r^2}{2}\frac{1}{r}\mathrm{d}x=
\frac{(r-1)^2}{2}\log(r-1)-\frac{(r-1)^2}{4}+\frac{1}{4}$
\end{center}
Hence,
\begin{center}
$\displaystyle \sum_{r > r' >0}r'\log{r'} =
\frac{(r-1)^2}{2}\log(r-1)-\frac{(r-1)^2}{4}+\frac{1}{4}+\frac{(r-1)\log(r-1)}{2}+\frac{B_2}{2!}\log(r-1)+\frac{B_4}{4!}(\frac{1}{(r-1)^2}-1)+R_m$,
\end{center}
We can easily show that,

$\displaystyle \frac{1}{r} \leq \log(r)-\log(r-1)$ and hence, $\displaystyle
\log(r-1)\leq \log(r)-\frac{1}{r} $ and we can upper bound the expression above
by
\begin{center}
$\displaystyle
\frac{(r-1)^2}{2}\log(r)-\frac{(r-1)^2}{4}+\frac{(r-1)\log(r)}{2}$
\end{center}

\end{proof}

\begin{myle}
(Approximation of a sum by a definite integral)

For an increasing and integrable on the summation domain function $f$, it holds
that:
\begin{equation}
\int_{s = a-1}^{b}f(s)\mathrm{d}s \leq \sum_{i = a}^{b} f(i) \leq \int_{s =
a}^{b+1}f(s)\mathrm{d}s
\end{equation}

Similarly, for a decreasing function $f$ it holds that:
\begin{equation}
\int_{s = a}^{b+1}f(s)\mathrm{d}s \leq \sum_{i = a}^{b} f(i) \leq \int_{s =
a-1}^{b}f(s)\mathrm{d}s
\end{equation}
\end{myle}

\addtocontents{toc}{\vspace{2em}}  
\backmatter

\label{Bibliography}
\lhead{\emph{Bibliography}}  
\bibliographystyle{alpha}  
\bibliography{Bibliography}  

\end{document}